%% file: main.tex
\documentclass[11pt]{article}
\usepackage{amsthm}
\usepackage{amsmath}
\usepackage{amssymb}
\usepackage{fullpage}
\usepackage{times}
\usepackage{paralist}
\usepackage[usenames]{color}
\usepackage{enumerate}
\usepackage{tabularx}
\usepackage{appendix}
\usepackage{graphicx}
\usepackage[boxed,vlined]{algorithm2e}
\usepackage[compact]{titlesec}
\usepackage{color}

\newcommand{\play}[2]{#1^{#2}}

\newcommand{\rand}[2]{#1^{#2}}
\newcommand{\plra}[3]{#1_{#2}^{#3}}

\renewcommand{\phi}{\varphi}

\newcommand{\set}[1]{\left\{#1\right\}}

\newcommand{\coloneq}{:=}

\newcommand{\hide}[1]{ }

\DeclareMathOperator*{\E}{\mathbb{E}}
\renewcommand{\phi}{\varphi}

\DeclareMathOperator{\embed}{embed}

\newcommand{\eps}{\epsilon}

\newcommand{\disj}{\text{\textsc{Disj}}}
\newcommand{\AND}{\text{\textsc{And}}}
\newcommand{\ta}{\text{\textsc{TaskAlloc}}}
\newcommand{\given}{\medspace | \medspace}
\DeclareMathOperator*{\MI}{I}

\DeclareMathOperator*{\SIC}{SIC}
\DeclareMathOperator*{\IC}{IC}
\DeclareMathOperator*{\CC}{CC}

\theoremstyle{plain}
\newtheorem{theorem}{Theorem}[section]

\newtheorem{lemma}[theorem]{Lemma}

\newtheorem{fact}[theorem]{Fact}

\newtheorem{definition}{Definition}

\renewcommand{\include}{\input}

\def\ShowAuthNotes{1}

\ifnum\ShowAuthNotes=1

\newcommand{\authnote}[2]{\textcolor{red}{\parbox{0.9\linewidth}{[{\footnotesize {\bf #1:} { {#2}}}]}}}

\else
\newcommand{\authnote}[2]{}
\fi

\newcommand{\Vnote}[1]{\authnote{Vnote}{#1}}

\title{Tight Bounds for Set Disjointness in the Message Passing Model}
\author{
Mark Braverman
\thanks{Department of Computer Science, Princeton University, 
{\tt mbraverm@cs.princeton.edu,} Research supported in part by an 
Alfred P. Sloan Fellowship, an NSF CAREER award (CCF-1149888), and a Turing
Centenary Fellowship.}\\
\and
Faith Ellen
\thanks{Department of Computer Science, University of Toronto,
{\tt faith@cs.toronto.edu,} Research supported in part by NSERC.}\\
\and
Rotem Oshman
\thanks{Department of Computer Science, University of Toronto,
{\tt rotem@cs.toronto.edu,} Research supported in part by NSERC.}\\
\and
Toniann Pitassi
\thanks{Department of Computer Science, University of Toronto,
{\tt toni@cs.toronto.edu,} Research supported in part by NSERC.}\\
\and
Vinod Vaikuntanathan
\thanks{Department of Computer Science, University of Toronto,
{\tt vinodv@cs.toronto.edu,} Research supported in part by NSERC, DARPA award FA8750-11-2-0225 and
an Alfred P. Sloan Fellowship.}\\
}

\begin{document}

\maketitle

\begin{abstract}
In a multiparty message-passing model of communication, there are $k$ players.
Each player has a private input, and they communicate by sending messages
to one another over private channels. While this model has been used
extensively in distributed computing and in multiparty computation,
lower bounds on communication complexity in this model and related models
have been somewhat scarce. In recent work~\cite{phillips12,woodruff12,woodruff13},
strong lower bounds of the form $\Omega(n \cdot k)$ were obtained
for several functions in the message-passing model; however, 
a lower bound on the classical Set Disjointness problem
remained elusive.

In this paper, we prove tight lower bounds
of the form $\Omega(n \cdot k)$ for the Set Disjointness problem
in the message passing model. Our bounds are obtained
by developing information complexity tools in the message-passing model,
and then proving an information complexity lower bound for Set Disjointness.
As a corollary, we show a tight lower bound for the task allocation problem~\cite{DruckerKuhnOshman}
via a reduction from Set Disjointness.
\end{abstract}

\newpage

\include{intro}
\include{overview}

\include{defs}
\include{directsketch}
\include{onebitsketch}
\include{marktemp}
\include{reductions}
%\bibliographystyle{plain}
%\bibliography{main}

\input{main.bbl}
%\include{model}

%\include{disj}
%\include{obstacles}
%\include{model}
%\include{directsum}
%\include{onebit}
%\include{useless}

%\include{reduction}

%\include{appendix}

\end{document}

%% file: intro.tex
\section{Introduction}
One of the most natural application domains for communication complexity is distributed computing:
When we wish to study the cost of computing in a network spanning multiple cores or physical machines,
it is very useful to understand how much communication is necessary,
 since communication
between machines often dominates the cost of the computation.
Accordingly,
lower bounds in communication complexity have been used to obtain many negative results in distributed computing,
from the round complexity of finding a minimum-weight spanning tree \cite{DHKKNPPW2012} to 
computing functions of distributed data \cite{PattShamir,KuhnOshman} and distributed computation and
verification of network graph structures and properties
\cite{DHKKNPPW2012,frischknecht12}.

To the best of our knowledge, however, all applications of communication complexity 
lower bounds in distributed computing to date
have used \emph{only two-player lower bounds}.
The reason for this appears to be twofold:
First,
the models of multi-party communication favored by the communication complexity community,
the \emph{number-on-forehead model} and the \emph{number-in-hand broadcast model},
do not correspond to most natural models of distributed computing.
Second, two-party lower bounds are surprisingly powerful,
even for networks with many players.
A typical reduction from a two-player communication complexity problem to a distributed problem $T$
finds a \emph{sparse cut} in the network, and shows that, to solve $T$,
the two sides of the cut must implicitly solve, say, Set Disjointness \cite{KushilevitzNisan}.
However, there are
problems that cannot be addressed by reduction from a two-player problem,
because such reductions must reveal almost the entire structure of the
network to one of the two players.
(One such example is described in~\cite{KuhnOshman}.)

In this paper, we study communication complexity
in \emph{message-passing models},
where each party has a private input, and the parties communicate by sending messages to each other over private channels.
These models have been used extensively in distributed computing,
for example, to study gossiping protocols~\cite{karp00},
to compute various functions of distributed data~\cite{kempe03},
and to understand fundamental problems,
such as achieving consensus in the presence of failures~\cite{FLP}.
Message passing models are also used to study privacy and security in 
multi-party computation.
% [VV] I don't like this sentence. It comes way too early and it's not entirely correct.
%
%To date, communication complexity has not found much application in these areas.
%We hope that our results will open a new avenue towards understanding privacy and the cost of communication
%in distributed multi-party computation.

In this paper, we choose to focus on the Set Disjointness problem \cite{story-setj}.
In Set Disjointness, denoted $\disj_{n,k}$, $k$ players each receive a set $X_i \subseteq [n]$,
and their goal is to determine whether the intersection $\bigcap_{i=1}^k X_i$
is empty or not.
An $\Omega(n)$
lower bound on the
two-player version of Set Disjointness,
due to Kalyanasundaram, Schnitger and Razborov \cite{KalyanasundaramS92,Razborov92},
is one of the most widely applied lower bounds in communication complexity.
The lower bound was recently re-proven as an information complexity lower bound \cite{baryossef04},
showing that any protocol for two-party set disjointness must ``leak'' a total of $\Omega(n)$ bits about the input.
%Here we prove an $\Omega(nk)$ lower bound on $\disj_{n,k}$.

Our main result is a tight lower bound on the communication complexity of the set
disjointness problem in a multiparty message-passing model, namely the coordinator model of 
Dolev and Feder~\cite{DF92}. 
This lower bound implies a corresponding bound in the ``truly distributed'' message-passing model, where there is no coordinator.
Our main technical tool in this paper is \emph{information complexity}, which has 
its origins in the work of Chakrabarthi, Shi, Wirth and Yao~\cite{CSWY01}, and which has 
recently played a pivotal role in several communication complexity lower bounds.

Our main theorem is an $\Omega(nk)$ lower bound following lower bound on the informaiton
complexity (and hence also the communication complexity) of the
set disjointness function in the multi-party coordinator model.

\begin{theorem}
For every $\delta > 0$, $n \geq 1$ and $k = \Omega(\log n)$, 
there is a distribution $\zeta$ such that 
the information complexity of Set Disjointness
is
$IC_{\zeta,\delta}(DISJ_{n,k})  = \Omega(nk)$  and
its communication complexity 
is
$CC_{\delta}(DISJ_{n,k}) = \Omega(nk)$.
\end{theorem}

We then apply this lower bound to obtain a lower bound of $\Omega(nk)$ on the
Task Allocation problem, $\ta_{n,k}$.
In this recently proposed problem \cite{DruckerKuhnOshman},
$k$ players must partition $n$ tasks among themselves.
Task Allocation is a useful primitive for distributed systems, where
a number of tasks
%often some amount of work
must be performed
by the participants in the computation,
but not every participant is able to carry out every task.
We describe this problem more formally below.

%In turn, the lower bound on Task Allocation implies a lower bound of $\Omega(n^2)$ on the number of bits 
%that must be communicated to find any rooted spanning tree in a directed network.

\paragraph{Information complexity and communication complexity.}
Our main technical tool in this paper is \emph{information complexity}.
The main technical result of the present paper concerns the problem of set disjointness.
Variants of set disjointness are perhaps the most studied problems in
communication complexity. 
In the two-party case, it is not hard to see that evaluating the disjointness of two subsets 
of $[n]$ {\em deterministically} requires at least $n+1$ bits of communication,
for example, using a fooling 
set argument \cite{KushilevitzN97}. In the randomized model, when error is allowed,
an $\Omega(n)$ lower 
bound is also known, although it is considerably more difficult to prove \cite{KalyanasundaramS92,Razborov92}. 
This result was later improved using information-theoretic techniques by Bar-Yossef et al. \cite{baryossef04}. 
Further advances in information complexity allow one to calculate the two-party communication complexity 
of disjointness precisely, up to additive $o(n)$ terms \cite{BGPW12}. 

In the multi-party case, there are three main models to consider, all with interesting applications. 
The first model is the {\em number on forehead (NOF)} model, where each player is given all inputs except for one. 
The NOF model has important connections to circuit lower bounds for the $\mathbf{ACC^0}$ class \cite{beigel1991acc}.
Since the disjointness problem has small $\mathbf{AC^0}$ circuits, this means that for $k>\log n$, the communication 
complexity of NOF disjointness is polylogarithmic. Also notice that since the entire calculation in this 
case can be carried out by two participants, yielding a trivial $O(n)$ upper bound. The exact dependence 
of the communication complexity on $n$ and $k$ has been the subject of considerable investigation \cite{chattopadhyay2008multiparty,lee2008disjointness,sherstov2012multiparty}, 
with the currently strongest lower bound being $\Omega(\sqrt{n}/2^k k)$ \cite{sherstovcommunication}. The second model is the {\em number in hand blackboard model}.
 In this model each party is only given her input, and the communication is carried out 
via a blackboard, so each message transmitted by a player is received by all other players. 
In this case, the communication complexity of disjointness might be as high as $\Theta(n\log{k})$ (note that an $\Omega(n)$ lower bound 
is trivial). Due to applications in streaming computation lower bounds, the version where the sets are either fully disjoint 
or have a single element in common has been studied. A lower bound of $\Omega(n/k)$ has been shown in this case 
\cite{chakrabarti2003near,gronemeier2009asymptotically,jayram2009hellinger} using information-theoretic techniques. The information complexity approach usually 
proceeds in two steps: first, a direct-sum result shows that the amount of information the players convey about 
the problem is additive about the coordinates of the problem (i.e. scales with $n$ in the case of disjointness); second, 
it is shown that to solve the one-coordinate version of the problem one needs to convey a non-trivial amount of information 
($\Omega(1/k)$ in the case of the blackboard model). 

In this paper, we consider {\it message-passing} models of communication complexity.
In all of the multi-party
models discussed so far, messages are {\it broadcast} to a
centralized blackboard, so that the entire communication transcript
is seen by all players. In message passing models (also known as
private channel models), the players communicate to one another
by sending and receiving messages through private pairwise channels.
Unlike the other models, it is possible to achieve $\Omega(n\cdot k)$ lower bounds
on problems in message passing models
\cite{DR98}.
We will focus on the {\em coordinator}
message passing model because it is most similar to standard
communication complexity models, and lower bounds in this model
imply similar lower bounds for other message passing models.
%known as the 
%{\em number in hand message-passing model}, where players can send messages to 
%each other. For technical reasons it is often more natural to replace this model with a 
%FAITH: As far as I can see, no paper defines a number in hand
%message passing model, except for \cite{phillips12},
%which does so in a very informal way.
%That paper is cited later when the simulation of a message passing
%protocol by a coordinator protocol is given.
In the {\em coordinator model} \cite{DF92}, the players
communicate with a coordinator by sending and receiving messages
on private channels.
%send through private channels with a coordinator. Note that any message-passing protocol can be simulated by a 
%coordinator protocol with a multiplicative $\log k$ overhead.
%Unlike the other models, 
%message passing model one can  hope to achieve an
%it is possible to achieve
%$\Omega(n\cdot k)$ lower bounds on problems in the coordinator model
%(which we do in this paper for disjointness).

A recent paper \cite{phillips12} developed
a new technique, called \emph{symmetrization}, for obtaining lower bounds of the form $\Omega(n \cdot k)$ 
via a reduction to the two party case.
%One promising 
%direction of work is obtaining lower bounds for the
%message-passing
%coordinator
%model via reductions to the two-party model using
%a technique called {\em symmetrization} \cite{phillips12}.
The symmetrization technique works for coordinate-wise problems such as Set Intersection,
where the parties need to compute the intersection 
of their sets; this amounts to coordinate-wise AND on the players' inputs.
However, symmetrization seems to fall short of yielding results 
for the multiparty Set Disjointness problem, and the development of new information-theoretic machinery
seems necessary.

Another recent line of work dealing with communication complexity in the message-passing setting appears in~\cite{woodruff12,woodruff13}.
In these papers, the main interest is in \emph{distributed streaming} or \emph{distributed data aggregation}:
each of $k$ machines holds some data set or receives an input stream,
and we wish to compute or approximate some function of the joint input,
either through a central coordinator~\cite{woodruff12} or in a decentralized manner~\cite{woodruff13}.
In~\cite{woodruff12}, a lower bound of $\Omega(n \cdot k)$ is proven for the Gap-Majority(2-DISJ) problem:
here the coordinator holds a set $S$, each player $i \in [k]$ holds a set $T_i$,
and the goal is to distinguish the case where a ``large majority'' of the intersections $\set{ \disj_{n,2}(S, T_i)}_{i=1}^k$
are empty from the case where only a ``small minority'' are empty.
(The precise values of ``large majority'' and ``small minority'' are parameters to the problem.)
To obtain this lower bound,~\cite{woodruff12} first proves a direct-sum like result showing that,
in order to compute the Gap-Majority of $k$ bits $Z_1,\ldots,Z_k$ in the message-passing model,
a protocol must leak $\Omega(k)$ bits of information about $Z_1,\ldots,Z_k$.
The known $\Omega(n)$ lower bound on the information complexity of $\disj_        {n,2}$~\cite{baryossef04}
is then applied to ``lift'' the $\Omega(k)$ one-bit lower bound on Gap-Majority
to an $\Omega(nk)$ lower bound on Gap-Majority(2-DISJ).
In~\cite{woodruff13}, similar techniques are used to obtain optimal $\Omega(nk)$ lower bounds on
a variety of problems in the decentralized message-passing model, including computing the number of distinct elements
in the joint input, and checking various graph properties when the input is interpreted as a graph.

%An additional complication is that, as discussed below, it is not necessarily true that the 
%players learn a lot of information about other players' inputs.
%We have to carefully account for the {\em sum} of the information 
%learned by the players and by the coordinator.

\paragraph{The message-passing model and its history in distributed computing.}
The message-passing model is one of the fundamental models in the theory of distributed computing,
and many variations of it have been studied.
The famous consensus impossibility result~\cite{FLP}
was originally proven for message-passing systems with faulty processors,
and it is also a common setting for other forms of consensus  (e.g., Byzantine consensus~\cite{pease80} and randomized consensus~\cite{benor83},
among many others).
%consensus for
%partially-synchronous systems~\cite{dwork88}, and other interesting variants. TOO FAR AWAY.
Lamport's seminal paper introducing the causal order of events in a distributed system~\cite{lamport78}
and his later Paxos protocol for consensus~\cite{lamport98} are also situated in the message-passing model.
There is also much work on achieving data consistency through replication in message-passing systems (e.g.,~\cite{attiya90}).
More recently,
\emph{gossip protocols}~\cite{feige90,karp00,kempe03,pietro08,alistarh10} have received considerable attention.
In \emph{gossip} (also called \emph{rumor-spreading}), the goal is to quickly disseminate  information throughout the network
or compute some aggregate function of the information;
this is achieved by having every node contact a small number of other nodes (typically, but not always, selected at random) to exchange
information with them.
Gossip protocols often use very large messages; for instance, a node might forward all the rumors it has collected so far
in every round.
%To our knowledge, the \emph{communication complexity} of the gossip approach has so far not been studied,
%and it is not known whether the sparse connections used in this style of algorithm
%necessarily yield a large overhead in the message size.

%In a gossip (or \emph{rumor-spreading}) protocol, each node contacts a small number of other nodes (typically selected at random) in every round,
%and exchanges information with them.
%ROTEM: after commenting out the sentence above it is no longer clear what is a gossip algorithm and what distinguishes it from everything else. I'm putting it back, rephrased.

These problems and others are sometimes studied in fully-connected networks,
and sometimes in networks with an arbitrary graph topology,
where communication is more restricted. 
%The model we consider in this paper is the basic model
%where every node can directly communicate with every other node.
%It remains interesting future work to extend and apply our techniques in settings where communication is governed
%by an underlying network topology which is not fully-connected.
Lower bounds for the model that we focus on in this paper, the \emph{coordinator model}, 
implies lower bounds for the basic message passing model where
every node can directly communicate with every other node.
It remains interesting future work to extend and apply our techniques in settings
where communication is governed by an underlying network topology which
is not fully connected.

Finally, we point out that the coordinator model is interesting in itself:
although it does not model a fully-decentralized 
distributed system, it is appropriate for data centers or for sensor networks 
with centralized control.
There is a growing body of work on streaming and sketching algorithms
set in the coordinator model~\cite{cormode05,amit05,nelson12}.

\paragraph{Connection to secure multiparty computation.}
Our results also have applications to showing lower bounds on the ``amount of privacy'' 
that one can achieve
in the context of secure multiparty computation. 

In the field of secure multiparty computation, the goal is for
$k$ players to communicate over a network to compute a joint function $f$ on
their inputs $x_1,\ldots,x_k$ while ensuring that no coalition of $t$ players
learn any information about the remaining players' inputs (other than what is
already implied by their own inputs and outputs). 
In the 1980s, the work of Ben-Or,
Goldwasser and Wigderson~\cite{BGW} showed multiparty protocols in the message-passing model
for computing
any function in an information-theoretically private way, assuming that the 
corruption threshold $t < k/2$.\footnote{While our description focuses on the notion of semi-honest corruptions
where the adversary corrupts $t$ players who run the protocol as prescribed, but try to learn information about the other
players' inputs from the transcript of the protocol execution. These results have also been extended 
to provide strong notions of security 
against malicious
corruptions, sometimes at the expense of a smaller corruption threshold $t$.} 
In addition, we know that information-theoretic {\em perfect} privacy is impossible to achieve if $t \geq k/2$. 
That is, the adversary must learn some information about the honest players' inputs in this setting.
An important 
question that remains is: {\em how much information} must the parties reveal about their inputs
in order to compute a function $f$?

Recently, a number of works investigated this quantitative question {\em in the two-party setting} from the framework of 
information complexity~\cite{feigenbaum,toniswork}. We believe that the information complexity tools developed
here will lead to a better quantitative understanding of privacy in multiparty computation.
For example, our information complexity lower bound already shows that in any $k$-party protocol for set disjointness
there is a constant fraction of players $i$ for which either (a) player $i$ learns $\Omega(n)$ bits of information about the
collective inputs of the players in $[k] \setminus \{i\}$, or (b) player $i$ ends up revealing  $\Omega(n)$ bits of information
about its own input to the other players. We leave a more thorough investigation of this connection as future work.

%\Vnote{\bf TODO}
\hide{
A corollary of our main result on the information-complexity of set disjointness 
is that any multiparty protocol for set disjointness in the full-information 
message-passing model has to reveal either the entire input of a party to the 
adversary 
}

\paragraph{Organization of the Paper.}
The remainder of the paper is organized as follows. 
We begin by giving some intuition about
% why the set disjointness problem is hard and a high level discussion of
our approach
for obtaining an $\Omega(kn)$ lower bound on the communication
complexity of set disjointness.  
In Section \ref{prelim},  we present necessary definitions
and  facts about information theory, 
Hellinger distance, and information complexity.
The next two sections present our lower bound, first proving that
the information cost of solving $\disj_{n,k}$ is at least $n$ times
the information cost of solving $\disj_{1,k} = \AND_k$, and then proving that it is at least $\Omega(k)$. Finally, in Section \ref{task}, we
reduce set disjointness to the task allocation problem,
to obtain an $\Omega(kn)$ lower bound on its communication
complexity.

%% file: overview.tex
\section{Overview: Why is Set Disjointness Hard?}
\label{sec:overview}
Before diving into the technical details, 
let us explain the motivation behind our definition of information cost and the hard distribution we use in the lower bound.

\paragraph{Choosing the ``right'' notion of information complexity.}
There are several possible ways to quantify the amount of information leaked by a protocol that solves $\disj_{n,k}$,
which might at first glance seem natural:
\begin{itemize}
	\item \emph{External information cost}, $\MI( \mathbf{X} ; \Pi(\mathbf{X}))$: how much information an external observer gains about the input $X$
		by observing the transcript of all the players and the coordinator.
		External information cost was used to prove the optimal $\Omega(n/k)$ lower bound on Promise Set Disjointness
		in the broadcast model~\cite{gronemeier2009asymptotically}.

		The external information cost can also be viewed as the \emph{coordinator's information cost},
		because the coordinator observes the entire transcript and does not initially know any of the inputs.
	\item \emph{The players' information cost}, $\sum_i \MI(\mathbf{X}^{-i} ; \Pi^i(\mathbf{X}) \given \mathbf{X}^i)$:
		how much the players together learn about the input $X$
		from their interactions with the coordinator,
		given their private input.
\end{itemize}
Unfortunately, neither of these is high enough to yield an $\Omega(kn)$ lower bound on Set Disjointness.
It is easy to see that the players' information cost is not always high:
In the trivial protocol where all players send their inputs to the coordinator,
the players do not learn anything.
Of course, in this protocol, the coordinator learns the entire input.

Likewise, the coordinator's information cost is not always high.
To see why, consider the following protocol:
For each coordinate $j$,
the coordinator searches for the smallest index $i$ such that
$X^i_j = 0$, by contacting the players in order $i = 1, \ldots, k$
and asking them to send $X_j^i$.
If $X^i_j = 0$ for some $i$, then $j \not \in \bigcap_{i = 1}^k X^i$,
and the coordinator moves on to coordinate $j + 1$ without asking the remaining players $\ell > i$ for $X^{\ell}_j$.
Otherwise, all players $i \in [k]$ have $X^i_j = 1$) and the coordinator halts
with output ``no'', as $j \in \bigcap_{i=1}^k X^i$.

The transcript of the protocol can be losslessly compressed into $O(n \log k)$ bits
by simply writing, for each coordinate $j$, the index of the first player $i$ that has $X^i_j = 0$,
or writing 0 if there is no such player.
Therefore the coordinator cannot learn more than $O(n \log k)$ bits about the input by observing the transcript.
On the other hand, in this protocol the players gain a significant amount of information:
each player $i$ from which the coordinator requests $X_j^i$
learns that $X_j^{\ell} = 0$ for all $\ell < i$.
This is not necessarily a lot of information.
In fact, in the distribution we design below,
it will correspond to roughly one bit of information,
However, it is learned by \emph{many players}.
Because each message is sent to only one player,
and we are interested in the \emph{total}
amount of communication between the coordinator and the players,
we can separately charge each player that learns this
bit of information, as this requires the coordinator to communicate
separately with each of them.

As we have seen, there is a protocol where the players learn nothing,
but the coordinator learns a lot,
and there is a protocol where the coordinator learns very little,
but the players learn a lot.
We will show that this trade-off is inherent,
by bounding from below the sum of the information learned
by the coordinator about the players' inputs
and the information learned by each player
from the coordinator (about the inputs of the other players).

\paragraph{Designing a hard distribution.}
From the example above, we see that a hard distribution should
make it hard for the coordinator to find the players that have zeroes,
forcing it to communicate with $\Omega(k)$ players
about each coordinate $j \in [n]$.
This means that with reasonably large probability, in each coordinate $j$
there should only a few players that have $X_j^i = 0$.
On the other hand, our distribution should have \emph{high entropy},
because, otherwise, the players can use Slepian-Wolf coding~\cite{slepianwolf}
to convey their joint input $X$ to the coordinator using roughly $O(H(X))$ bits.
In order to balance these two concerns, we follow~\cite{baryossef04},
and use a \emph{mixture of product distributions}.

Our hard distribution is a product $\eta = \xi^n$,
where $\xi$ is a hard distribution for a single coordinate $j \in [n]$.
Informally, $\xi$ has two ``modes'', selected by a ``switch'' $\mathbf{M}_j \in \set{0,1}$:
\begin{itemize}
\item An ``easy'' mode, $\mathbf{M}_j = 0$, where each $\mathbf{X}_j^i = 0$ with probability $1/2$ independently.
\item A ``hard'' mode, $\mathbf{M}_j = 1$, where there is exactly one player $i$ with $\mathbf{X}_j^i = 0$, and the remaining players $\ell \neq i$
		have $\mathbf{X}_j^{\ell} = 1$.
		The identity of the player that receives a zero is a random variable $\mathbf{Z} \in_{\mathsf{U}} [k]$.
\end{itemize}
More formally, for each $j \in [n]$,
there is an independent distribution $\xi$ over triples
$(\mathbf{X}_j, \mathbf{M}_j, \mathbf{Z}_j)$,
where $\mathbf{X}_j \in \set{0,1}^{k}$, $\mathbf{M}_j \in \set{0,1}$, and
$\mathbf{Z}_j \in [k]$,
such that the components
$\mathbf{X}_j^1,\ldots,\mathbf{X}_j^k$ of $\mathbf{X}_j$
are independent given $\mathbf{M}_j$ and $\mathbf{Z}_j$.
Each player $i$ is given the input $\mathbf{X}^i_1,\ldots,\mathbf{X}_n^i$.

It may seem surprising that, under our distribution $\eta$,
the answer to Set Disjointness is {\em almost always} ``yes'':
The probability that we get some coordinate
$j \in \bigcap_{i = 1}^n \mathbf{X}^i$ is roughly $n/2^k$, which is negligible
when $n$ is significantly larger than $2^k$.
This is necessary for our direct sum theorem (see below).
However, it might seem to make $\eta$
an easy distribution, rather than a hard one.
The key to $\eta$'s hardness lies in the fact that the protocol
must succeed with high probability on \emph{any} input,
even inputs that are very unlikely under $\eta$.
This means that for hard coordinates,
the protocol must ``convince itself'' that there really is
some player that had a zero.
This is hard because it is difficult to find such a player.

\paragraph{Ruling out Slepian-Wolf coding.}

As observed in~\cite{phillips12} and as mentioned above, any lower bound for Set Disjointness (or in the case of~\cite{phillips12}, bitwise-OR and other bitwise functions)
must implicitly rule out an approach where the players use Slepian-Wolf or other clever coding techniques to convey their inputs to the coordinator efficiently.
Our lower bound does this quite explicitly.

Under the distribution $\eta = \xi^n$, we think of the players as jointly ``owning'' the input $\mathbf{X}$,
because they are the only ones that initially know it.
On the other hand, we think of the coordinator as ``owning'' the switches, $\mathbf{M} = \mathbf{M}_1,\ldots,\mathbf{M}_n$:
the coordinator can easily determine if a given coordinate is ``easy'' or ``hard'' by sampling
$O(\log n)$ players' inputs---if it finds no zeroes, it can conclude that the coordinate is ``hard'' with very high probability (in $n$).
Since we are aiming for an $\Omega(nk)$ lower bound and the coordinator can determine $\mathbf{M}$ using $O(n \log n)$
bits, we may as well give this information to the coordinator for free.

Given that a coordinate $j$ is hard, its entropy is only $1/k$.
If the coordinator could convey the set of hard coordinates
(or enough information about this set) to the players,
they could then use Slepian-Wolf coding
to send this part of the input to the coordinator in roughly $O(n)$ total bits
(one bit per hard coordinate).
However, the entropy of the set of hard coordinates is $n/2$, so 
conveying it (or sufficient information about it) to the players requires
the coordinator to send $\Omega(n)$
bits to each player, for a total of $\Omega(nk)$ bits.
In the absence of this information, the overall entropy of the input
is $\Omega(nk)$, ruling out this type of approach.

We will formalize this intuition by showing that any protocol for
Set Disjointness is ``bad'' in one (or both) of the following ways.
\begin{enumerate}[(1)]
	\item The players convey to the coordinator ``useless'' information about their inputs:
		in the easy case
		when $\mathbf{M}_j = 0$,
		the coordinator learns $\Omega(k)$ bits about 
		coordinate $j$, $\mathbf{X}_j^1,\ldots,\mathbf{X}_j^k$.
		This information is ``useless'' for the coordinator because when $\mathbf{M}_j = 0$
		it can safely ignore coordinate $j$:
		with overwhelming high probability
		the sets do not intersect there.

		One example of this approach is the naive protocol where players send their entire input to the coordinator.
	\item If the players do not convey to the coordinator a lot of information when $\mathbf{M} = 0$,
		then we will show that the coordinator conveys to the players ``useless'' information about the set of hard coordinates:
		$\Omega(k)$ players must learn whether coordinate $j$ is easy
		(more formally, they learn $\Omega(1)$ bits of information about coordinate $j$)
		even when their input is $\mathbf{X}_j^i = 1$,
		i.e., they are not the special player that the coordinator is searching for.

		An example of this approach is the protocol where the coordinator first samples a few inputs
		to determine which coordinates are hard, then sends the set of hard coordinates to all the players;
		each player responds by sending the coordinator a list of the hard coordinates 
		where its input is zero.
\end{enumerate}
In our lower bound proof, we explicitly bound from below the sum of the information costs described above.

%% file: defs.tex
\section{Preliminaries}
\label{prelim}
\label{sec:prelim}

\paragraph{Notation.}
We use boldface letters to denote random variables,
and capital letters to denote vectors or sets.
For a set $A \subseteq [k]$, we let $\bar{e}_A$ denote the complement of $A$'s characteristic vector;
that is, $\bar{e}_A$ has 1 in exactly those coordinates that are not elements of $A$.
For convenience we drop the curly brackets, so that, for example, $\bar{e}_{i,j} = \bar{e}_{\set{i,j}}$.

If $X \in \set{0,1}^{k \cdot n}$ is a $k$-tuple of $n$-bit inputs,
then $X^i \in \set{0,1}^n$ denotes the input to the $i$-th player, 
and $X^i_j \in \set{0,1}$ denotes the $j$-th coordinate of $X^i$.
For an $n$-tuple $Y \in \set{0,1}^n$, we use 
$$Y_{-i} = Y_1,\ldots,Y_{i-1},Y_{i+1}, \ldots,Y_n$$
to denote the tuple obtained from $Y$ by dropping the $i$-th coordinate.
We also let $Y_{[i,j]} \coloneq Y_i,\ldots,Y_j$.
Finally, $\embed(X, i, x)$ denotes the vector obtained from $X$ by inserting $x$ in coordinate $i$:
$\embed(X, i, x) = (X^1,\ldots,X^{i-1},x,X^i,\ldots,X^m)$ where $|m| = |X|$.

\paragraph{Models of computation.}
%We consider an asynchronous message passing model.
%There are $k$ players, $1,\ldots, k$, and there is
%a private channel between every pair of players.
%Each player has an independent source of private random bits.
%Player $i$ intially has a private input $X^i$.
%The goal is to compute some function $f(X^1,\ldots,X^k)$,
%perhaps with some small probability of error,
%and have player 1 output the answer.
%Each player can send messages to any other player.
%We assume that all players and all channels are reliable.
%In particular, all messages that are sent on a channel will eventually
%be received in the same order that they were sent.
%At the end of a protocol, all players know the output.
%Our complexity measure is the total number
%of bits sent on all channels.

As mentioned in the Introduction, we will work in the asynchronous
{\em coordinator} message passing model introduced in \cite{DF92}. In this model, there is one additional
participant, called the coordinator, which
receives no input, and there is a private channel between
every player and the coordinator.
However, there are no channels between the players,
so they cannot communicate directly with one another.
The coordinator also has a private source of randomness.
On each channel, the messages alternate between the coordinator 
and the player. Messages are required to be self-delimiting,
so both the coordinator and the player know when one message
(from coordinator to the player or vice-versa) has been completely sent. 
Each player $i$ knows whether or not it is his/her turn to speak
by looking at the transcript $\Pi_i$ between player $i$ and the
coordinator. If the last message sent in this transcript was from
the coordinator, then it is player $i$'s turn to speak. 
The coordinator can communicate whenever he is no longer
waiting for anyone to speak. 
This happens when in each transcript
$\Pi_j$, $j \leq k$, the last message sent in this transcript was from
the player. When it is the coordinator's turn to speak,
he can send messages to as many players as he wishes.
At the end of a protocol, the coordinator outputs the answer.
Our complexity measure is the {\it total} number of bits sent
on all channels.
Since our complexity measure is the total number of bits,
we can assume without loss of generality that the model is sequential
and round based: in the first round, the coordinator speaks to
exactly one player, and in the next round, this player responds, and so on.
%The coordinator model is defined formally in the appendix, and we compare it to
%various message passing models in the literature.

%The coordinator model can be simulated by the message
%passing model with no overhead by letting player 1
%also play the role of coordinator.
%Conversely, as observed in \cite{philipps12},
%any protocol in the message passing model
%can be simulated in the coordinator model by having
%each player $i$ that wants to send a message $m$ to player $j$
%instead send $j;m$ to the coordinator, who 
%forwards the message $i;m$ to player $j$.
%This increases the number of bits communicated
%during the protocol by a factor at most $2(1 + \lceil \log_2 \rceil)$.

For any protocol $\Pi$ and any input $X \in \set{0,1}^{k \cdot n}$,
we let $\Pi(X)$ denote the distribution of $\Pi$'s transcript
(as seen by the coordinator)
when run with input $X$,
and, for each player $i \in [k]$, we let $\Pi^{i}(X)$
denote the transcript of messages 
sent between player $i$ and the coordinator (in both directions).

\paragraph{Communication complexity.}
Let $\Pi$ be a protocol for solving a problem $\mathcal{P}$.
The \emph{error} of $\Pi$ is given by
	\begin{equation*}
		\max_X \Pr\left[ \text{player 1 outputs an incorrect answer} \right],
	\end{equation*}
where the probability is taken over the private randomness of the 
coordinator and the players.

The \emph{communication complexity} of a protocol $\Pi$
is the maximum over all inputs $X$
of the maximum number of bits exchanged between the players and the coordinators when $\Pi$ is executed with input $X$.
The \emph{$\delta$-error randomized communication complexity} of a problem $\mathcal{P}$ in the coordinator model,
which we denote  by $\CC_{\delta}(\mathcal{P})$, is the minimum communication complexity of any randomized protocol $\Pi$
that solves $\mathcal{P}$ with error at most $\delta$.

\paragraph{Useful classes of distributions.}
Let us define the class of distributions we use for our direct sum theorem and the lower bound for 1-bit AND.
Fix an \emph{input domain} $\mathcal{X} = \mathcal{X}_1 \times \ldots \times \mathcal{X}_k$,
and let $\mathcal{X} = (\mathcal{X}^1,\ldots,\mathcal{X}^k)$ be a random variable denoting the input.
Our hard distribution uses an auxiliary ``switch'' $\mathbf{M}$, which determines if a coordinate is hard or easy,
and another auxiliary variable $\mathbf{Z}$, which selects the player that receives zero in the hard case.
Conditioned on $\mathbf{M}$ and $\mathbf{Z}$, the players' inputs are independent from each other.
The value of $\mathbf{M}$ is assumed known to the coordinator, but the value of $\mathbf{Z}$ is hidden from all participants.

The following definition captures distributions that behave like our hard distribution.
It is a special case of a mixture of product distributions~\cite{baryossef04}.
\begin{definition}[Switched distributions]
	We say that the joint distribution $\eta$
	of $(\mathbf{X}, \mathbf{M}, \mathbf{Z})$
	is \emph{switched by $\mathbf{M}$ and $\mathbf{Z}$} if $\mathbf{X}^1,\ldots,\mathbf{X}^k$ are conditionally independent
	given $\mathbf{M}$ and $\mathbf{Z}$, and $\mathbf{M}$ is independent from $\mathbf{Z}$.
\end{definition}

Our hard distribution for a single coordinate also has the property that with very high probability, it produces a Set Disjointness instance on which the answer is ``yes''. This is important for our direct sum reduction.
Adapting the definition of a \emph{collapsing distribution} from~\cite{baryossef04}, we capture this notion as follows.
(The following definition is specifically for 1-bit AND; it is easy to generalize to arbitrary functions along the same lines as~\cite{baryossef04}.)
\begin{definition}[$\eps$-collapsing distributions]
	We say that a distribution $\mu : \set{0,1}^{k} \rightarrow [0,1]$ is \emph{$\eps$-collapsing for AND}
	if
	\begin{equation*}
		\Pr_{\mathbf{X} \sim \mu}\left[ \bigwedge_{i=1}^k \mathbf{X}_k = 1  \right] \leq \eps.
	\end{equation*}
\end{definition}

\paragraph{Information theory and Hellinger distance.}

Let $\mu$ be a distribution on a finite set $D$ and let $X, Y, Z$ be
random variables.
The {\it entropy} of $X$ is defined by
$$H(X) = \sum_{\omega \in D} \mu(\omega) \log \frac{1}{\mu(\omega)}$$

\noindent The {\it conditional} {\it entropy} of $X$ given $Y$ is
$$H(X|Y) = \sum_y H(X | Y=y)Pr[Y=y],$$
where $H(X|Y=y)$ is the entropy of the conditional distribution of $X$
given the event $\{Y=y\}$.

\noindent The {\it joint} {\it entropy} of $X$ and $Y$ is the entropy of their joint
distribution and is denoted by $H(X,Y)$.

\noindent The {\it mutual} {\it information} between $X$ and $Y$ is
$$I(X;Y) = H(X) - H(X|Y) = H(Y) - H(Y|X).$$

\noindent The {\it conditional} {\it mutual} {\it information} between $X$ and $Y$ conditioned on $X$ is
$$I(X;Y|Z) = H(X|Z) - H(X|Y,Z).$$

\noindent The {\it Hellinger} {\it distance} between probability distributions $P$
and $Q$ on a domain $\mathcal{D}$ is defined by
$$h(P,Q) = \frac{1}{\sqrt{2}} \sqrt{ \sum_{\omega \in \mathcal{D}} | \sqrt{P(\omega)} - \sqrt{Q(\omega)}|^2 }.$$

\noindent The square of the Hellinger distance is:
$$h^2(P,Q) = 1 - \sum_{\omega \in D} \sqrt{P(\omega)Q(\omega)}.$$

\noindent Hellinger distance is a metric and, in particular,
it satisfies the triangle inequality.
Another useful property of the Hellinger distance is the following:
\begin{lemma}[\cite{baryossef04}]
	Let $\mathcal{P}$ be a problem, and let $\Pi$ be a $\delta$-error protocol for $\mathcal{P}$.
	If $X$ and $Y$ are inputs such that $\mathcal{P}(X) \neq \mathcal{P}(Y)$,
	then $h(\Pi(X), \Pi(Y)) \geq (1 - \delta)/\sqrt{2}$.
	\label{lemma:h_delta}
\end{lemma}
Essentially, the lemma asserts that since the protocol must distinguish between the two inputs $X$ and $Y$,
the Hellinger distance of the respective distributions on the transcript must be large.

The following facts will be useful to us in the sequel:

\begin{fact}[Chain rule for mutual information~\cite{coverthomas}]
	For any $A_1,\ldots,A_n$, $B$ and $C$ we have
	\begin{equation}
		\MI( A_1 \ldots A_n ; B \given C) = \sum_{i=1}^{n} \MI(A_i ; B \given A_1 \ldots A_{i-1} C). 
		\label{eq:chainrule}
	\end{equation}
\end{fact}

\begin{lemma}[``Simplified chain rule'']
	If $A$ and $B$ are independent given $D$,
	then
	$\MI(A ; BC \given D) = \MI(A ; C \given B,D)$.
	\label{lemma:chain}
\end{lemma}
\begin{proof}
By the chain rule, $\MI(A ; BC \given D) = \MI(A ; B \given D) + \MI(A ; C \given B,D)$.
Since $A$ and $B$ are independent conditioned on $D$, we have $\MI(A ; B \given D) = 0$, and the claim follows.
\end{proof}

\begin{lemma}[\cite{braverman11}]
	If $A,B$ are independent given $D$, then $\MI(A ; C \given B,D) \geq \MI(A ; C \given D)$.
	\label{lemma:drop}
\end{lemma}
\begin{lemma}[\cite{baryossef04}]
	Let $\mu_0, \mu_1$ be two distributions.
	Suppose that $\mathbf{Y}$ is generated as follows:
	we first select $\mathbf{S} \in_{\mathsf{U}} \set{0,1}$,
	and then sample $\mathbf{Y}$ from $\mu_{\mathbf{S}}$.
	Then 
		$\MI( \mathbf{S} ; \mathbf{Y} ) \geq h^2(\mu_0, \mu_1)$.
	\label{lemma:MI}
\end{lemma}

\paragraph{Information cost.}
In general, we define the \emph{internal information cost} of a protocol as follows.
\begin{definition}
	Let $\mathbf{X} \sim \zeta$ be a distribution.
	The \emph{internal information cost} of a protocol $\Pi$ with $k$ parties communicating through a coordinator 
	with respect to $\zeta$
	is given by
	\begin{equation*}
		\IC_{\zeta}(\Pi)
		\coloneq
			\MI_{\mathbf{X} \sim \zeta}( \mathbf{X} ; \Pi(\mathbf{X}) ) + 
		\sum_{i \in [k]} \left[
			\MI_{\mathbf{X} \sim \zeta}( \mathbf{X}^{-i} ; \Pi^{i}(\mathbf{X}) \given \mathbf{X}^i)
		\right].
	\end{equation*}
	If $\mathcal{P}$ is a problem (formally, a Boolean predicate on $k \times n$-bit inputs and outputs from some
	domain),
	then we define the information complexity of $\mathcal{P}$ as 
	\begin{equation*}
		\IC_{\zeta, \delta}(\mathcal{P}) = \inf_{\Pi} \IC_{\zeta}(\Pi)
	\end{equation*}
	where the infimum is taken over all $\delta$-error randomized protocols for $\mathcal{P}$.
	\label{def:IC}
	\label{def:CIC}
\end{definition}
This is a general definition which does not depend on the structure of the distribution $\zeta$.
However, our lower bound uses a switched distribution, and as we explained in Section~\ref{sec:overview},
we give a bound on the following, more fine-grained expression:
\begin{definition}
	Let $(\mathbf{X}, \mathbf{M}, \mathbf{Z}) \sim \eta$ be a distribution switched by $\mathbf{M}$ and $\mathbf{Z}$.
	The \emph{switched information cost} of a protocol $\Pi$
	with respect to $\mu$
	is given by
	\begin{equation*}
		\SIC_{\eta}(\Pi)
		\coloneq
		\sum_{i \in [k]} \left[
		\MI_{(\mathbf{X}, \mathbf{M}, \mathbf{Z}) \sim \eta}( \mathbf{X}^i ; \Pi^{i}(\mathbf{X}) \given \mathbf{M}, \mathbf{Z})
		+
		\MI_{(\mathbf{X}, \mathbf{M}, \mathbf{Z}) \sim \eta}( \mathbf{M} ; \Pi^{i}(\mathbf{X}) \given \mathbf{X}^i, \mathbf{Z})
		\right].
	\end{equation*}
	The \emph{switched information cost} of a problem $\mathcal{P}$ is defined analogously.
	\label{def:SIC}
\end{definition}

In Sections~\ref{sec:directsum} and~\ref{sec:onebit} we show that the switched information cost of $\disj_{n,k}$
under our hard distribution is $\Omega(nk)$. In Section~\ref{sec:mark}
we use this fact to show that the internal information cost of $\disj_{n,k}$
is also $\Omega(nk)$.

To obtain a lower bound on the communication cost of a problem $\mathcal{P}$, it is sufficient to give a lower bound on its
internal information cost (or similarly, on its switched information cost):
\begin{lemma}\label{lem:cc-ic}
	For any problem $\mathcal{P}$, $\CC_{\delta}(\mathcal{P}) \geq 1/2 \cdot \IC_{\zeta, \delta}(\mathcal{P})$.
	\label{lemma:cc_ic}
\end{lemma}
\begin{proof}
	For any $\delta$-error protool $\Pi$,
	\begin{align*}
		\IC_{\zeta}(\Pi)
		&=
			\MI_{\mathbf{X} \sim \zeta}( \mathbf{X} ; \Pi(\mathbf{X}) ) + 
		\sum_{i \in [k]} \left[
			\MI_{\mathbf{X} \sim \zeta}( \mathbf{X}^{-i} ; \Pi^{i}(\mathbf{X}) \given \mathbf{X}^i)
		\right]
		\\
		&\leq
		H(\Pi) + \sum_{i \in [k]} H(\Pi^i \given \mathbf{X}^i)
		\\
		&\leq H(\Pi) + \sum_{i \in [k]} H(\Pi^i)
		\leq |\Pi| + \sum_{i \in [k]} |\Pi^i| = 2|\Pi|.
	\end{align*}
	The claim follows.
\end{proof}

\paragraph{Problem statements.}
In the Set Disjointness problem, $\disj_{n,k}$, each player receives an input $X^i \in \set{0,1}^n$,
and the goal is to compute
\begin{equation*}
	\disj_{n,k}(X^1,\ldots,X^k) = \bigvee_{j = 1}^n \bigwedge_{i = 1}^k X^i_j.
\end{equation*}

We also consider the \emph{Task Allocation Problem}, $\ta_{n,k}$.
Here we think of the elements $\set{1,\ldots,n}$ as \emph{tasks} that need to be performed.
Each player receives an input $X^i \subseteq [n]$ representing the set of tasks it is able to perform,
and the coordinator must output an \emph{assignment} $Y : [n] \rightarrow [k]$,
such that for each $j \in [n]$, $j \in X^{Y(j)}$; that is, every task is assigned to a player that
had that task in its input.

%% file: directsketch.tex
\section{Direct Sum Theorem}
\label{sec:directsum}
We begin by proving that the information cost of computing the set disjointness function
\[\disj_{n,k}(\mathbf{X}^1,\ldots,\mathbf{X}^k) = \bigvee_{j = 1}^n \bigwedge_{i=1}^k \mathbf{X}_j^i \]
is as least $n$ times the cost of solving the one-bit problem $\AND_{k} = \bigwedge_{i = 1}^k \mathbf{X}^i_j$.
The proof is by reduction: given a protocol $\Pi$ for $\disj_{n,k}$ and a switched distribution $\eta = \xi^n$,
where $\xi$ itself is a switched and $\eps$-collapsing distribution, we will construct a protocol
$\hat{\Pi}$ for $\AND_k$, such that  $\SIC_{\xi}(\hat{\Pi}) \leq (1/n) \SIC_{\eta}(\Pi)$.

The one-bit protocol $\hat{\Pi}$ uses $\Pi$ by
constructing an $n$-bit input, running $\Pi$ on it, and returning $\Pi$'s answer.
However, the input to $\hat{\Pi}$ is only a single bit per player.
To construct an $n$-bit input, the coordinator first selects a random coordinate $\mathbf{j} \in_{\mathsf{U}} [n]$,
into which the one-bit input to $\hat{\Pi}$ will be embedded.
Next we wish to randomly sample the other coordinates $[n] \setminus \set{ \mathbf{j}}$ from $\xi^{n-1}$,
in order to obtain an $n$-bit input on which we can run $\Pi$.
We must do this carefully:
we need $\hat{\Pi}$ to have an information cost proportionate to the information cost of $\Pi$, but
we do not know where $\Pi$ incurs the majority of its information cost---does
the coordinator learn a lot about the inputs given the switch $\mathbf{M}$, or do the players learn a lot about the switch $\mathbf{M}$ given their inputs?
One of these terms may be \emph{small}, and we must ensure that $\hat{\Pi}$'s corresponding cost in the same term is also small.
\begin{itemize}
	\item
		If in $\Pi$ the coordinator does not learn much about the input given $\mathbf{M}$ and $\mathbf{Z}$,
then our new protocol $\hat{\Pi}$ should also not reveal too much about the input to the coordinator.
A good solution is to have the coordinator sample $\mathbf{M}^{-\mathbf{j}}$
and $\mathbf{Z}^{-\mathbf{j}}$
and send them to the players, who can then sample their inputs independently using their private randomness.
\item If in $\Pi$ the players do not learn much about $\mathbf{M}$ given their inputs and $\mathbf{Z}$,
	then we should not reveal $\mathbf{M}$ to the players in $\hat{\Pi}$.
	A good solution is to have the coordinator sample $\mathbf{M}^{-\mathbf{j}}, \mathbf{Z}^{-\mathbf{j}}$ and $\mathbf{X}^{-\mathbf{j}}$, and send to each player $i$ its input $\mathbf{X}_i^{-\mathbf{j}}$.
	Thus the players do not know $\mathbf{M}$ before they execute $\Pi$ (except what they can deduce from their inputs).
\end{itemize}
Since we do not know in advance how $\Pi$ behaves on the average coordinate,
our solution is to ``hedge our bets'' by using the first approach to sample the coordinates below $\mathbf{j}$, and the second approach
to sample the coordinates above $\mathbf{j}$.
More formally, on one-bit input $(\mathbf{U}, \mathbf{N}, \mathbf{S}) \sim \xi$, protocol $\hat{\Pi}$ works as follows:
\begin{enumerate}
	\item The coordinator samples a random coordinate $\mathbf{j} \in_{\mathsf{U}} [n]$
	and samples $\mathbf{Z}_{-\mathbf{j}} \in_{\mathsf{U}} [k]^{n-1}$,
	and sends them to all players.
	\item For each $\ell < \mathbf{j}$, the coordinator samples $\mathbf{M}_{\ell}$
			and sends it to all players.
			Each player $i$ then samples $\mathbf{X}_{\ell}^i$ from its marginal distribution 
			given $\mathbf{M}_{\ell}$ and $\mathbf{Z}_{\ell}$.
	\item For each $\ell > \mathbf{j}$, the coordinator samples $\mathbf{X}_{\ell}, \mathbf{M}_{\ell}$
			from their marginal distribution given $\mathbf{Z}_{\ell}$,
			and sends to each player $i$ its input $\mathbf{X}_{\ell}^i$.
	\item The participants simulate the execution of $\Pi$ using the joint input
			\begin{equation*}
				\embed(\mathbf{X}, \mathbf{j}, \mathbf{U})
				=
				\set{ (\mathbf{X}_1^i, \ldots, \mathbf{X}_{\mathbf{j}-1}^i, \mathbf{U}^i, \mathbf{X}_{\mathbf{j}+1}^i, \ldots, \mathbf{X}_n^i}_{i = 1}^k.
			\end{equation*}
	\item The coordinator outputs the value output by $\Pi$.
\end{enumerate}
The last step is the reason we require $\xi$ to be $\eps$-collapsing:
for each coordinate $\ell \neq \mathbf{j}$, with probability at least $1 - \eps$ we have $\bigwedge_{i = 1}^k \mathbf{X}_{\ell}^i = 0$.
By union bound, the probability that $\bigvee_{\ell \neq \mathbf{j}} \bigwedge_{i = 1}^k \mathbf{X}_{\ell}^i = 0$
is at least $1 - (n - 1)\eps$.
Whenever this occurs we have $\disj_{n,k}(\embed(\mathbf{X}, \mathbf{j}, \mathbf{U})) = \AND_k(\mathbf{U})$, that is,
if $\Pi$ succeeds then $\hat{\Pi}$ succeeds as well.
Therefore the error probability of $\hat{\Pi}$ is at most $n \eps + \delta$,
where $\delta$ is the error probability of $\Pi$.

The following lemma relates the information cost of $\hat{\Pi}$ to that of $\Pi$:
\begin{lemma}\label{lem:directsum}
	For each player $i \in [k]$ we have
	\begin{align*}
		&\MI_{(\mathbf{U}, \mathbf{N}, \mathbf{S}) \sim \xi}\left( \mathbf{N} ; \hat{\Pi}^i(\mathbf{U}) \given \mathbf{U}^i, \mathbf{S} \right)
		\leq 
		\frac{1}{n}
		\left[
		\MI_{(\mathbf{X}, \mathbf{M}, \mathbf{Z}) \sim \eta} ( \mathbf{M} ;  \Pi^{i}( \mathbf{X}) \given \mathbf{X}^i, \mathbf{Z})
		\right]
		\qquad\qquad\text{and}
		\\
		&\MI_{(\mathbf{U}, \mathbf{N}, \mathbf{S}) \sim \xi}\left( \mathbf{U}^i ; \hat{\Pi}^i(\mathbf{U}) \given \mathbf{N}, \mathbf{S} \right)
		\leq 
		\frac{1}{n}
		\left[
		\MI_{(\mathbf{X}, \mathbf{M}, \mathbf{Z}) \sim \eta} ( \mathbf{X}^i ;  \Pi^{i}( \mathbf{X}) \given \mathbf{M}, \mathbf{Z})
		\right]
		.
	\end{align*}
	\label{lemma:reduction_sketch}
\end{lemma}
\begin{proof}
	We begin with the first inequality.
	For each player $i$,
	the player's view of the transcript of $\hat{\Pi}$ is given by
	\begin{equation*}
	\hat{\Pi}^i(\mathbf{U}) = \mathbf{j},\mathbf{Z}_{-\mathbf{j}},\mathbf{M}_{[1,\mathbf{j}-1]},\mathbf{X}^i_{[\mathbf{j}+1,n]},
	\Pi(\embed(\mathbf{X}, \mathbf{j}, \mathbf{U})).
\end{equation*}
By Lemma~\ref{lemma:chain}, since the tuple $\langle \mathbf{j}, \mathbf{Z}_{\mathbf{-j}}, \mathbf{M}_{[1, \mathbf{j}-1]}, \mathbf{X}^i_{[\mathbf{j}+1,n]} \rangle$
is independent from $\mathbf{N}$ conditioned on $\mathbf{U}^i$ and $\mathbf{S}$ (or even without the conditioning),
we can write
	\begin{align}
		&\MI_{(\mathbf{U}, \mathbf{N}, \mathbf{S}) \sim \xi}\left( \mathbf{N} ; \hat{\Pi}^{i}(\mathbf{U}) \given \mathbf{U}^i, \mathbf{S}\right)
		\nonumber
		\\
		&=
		\MI_{
		\substack{(\mathbf{U}, \mathbf{N}, \mathbf{S}) \sim \xi\\
		(\mathbf{X}_{-\mathbf{j}}, \mathbf{M}_{-\mathbf{j}}, \mathbf{Z}_{-\mathbf{j}}) \sim \xi^{n-1}
		}}
		\left( \mathbf{N} ; \mathbf{j}, \mathbf{Z}_{-\mathbf{j}},\mathbf{M}_{[1,\mathbf{j}-1]}, \mathbf{X}^i_{[\mathbf{j}+1,n]}, \Pi^{i}( \embed(\mathbf{X}_{-\mathbf{j}}, \mathbf{j}, \mathbf{U}) ) \given \mathbf{U}^i, \mathbf{S}\right)
		\nonumber
		\\
		&=
		\MI_{
		\substack{(\mathbf{U}, \mathbf{N}, \mathbf{S}) \sim \xi\\
		(\mathbf{X}_{-\mathbf{j}}, \mathbf{M}_{-\mathbf{j}}, \mathbf{Z}_{-\mathbf{j}}) \sim \xi^{n-1}
		}}
		\left( \mathbf{N} ;  \Pi^{i}( \embed(\mathbf{X}_{-\mathbf{j}}, \mathbf{j}, \mathbf{U})) \given
		\mathbf{j},\mathbf{M}_{[1,\mathbf{j}-1]}, \mathbf{X}^i_{[\mathbf{j}+1,n]}, \mathbf{U}^i, \mathbf{Z}_{-\mathbf{j}}, \mathbf{S}\right)
		\nonumber
		\\
		&=
		\MI_{
		(\mathbf{X}, \mathbf{M}, \mathbf{Z}) \sim \eta}
		\left( \mathbf{M}_{\mathbf{j}} ;  \Pi^{i}( \mathbf{X}) \given
		\mathbf{j},\mathbf{M}_{[1,\mathbf{j}-1]}, \mathbf{X}^i_{[\mathbf{j},n]}, \mathbf{Z}\right).
		\label{eq:Pihat}
	\end{align}
	Next, since $\mathbf{X}^i_{[1,\mathbf{j}-1]}$ and $\mathbf{M}_{\mathbf{j}}$ (which we previously called $\mathbf{N}$) are independent,
	even given the conditioning in~\eqref{eq:Pihat},
	we can apply Lemma~\ref{lemma:drop} to add conditioning on 
	$\mathbf{X}^i_{[1,\mathbf{j}-1]}$,
	yielding
	\begin{align*}
		&\MI_{(\mathbf{U}, \mathbf{N}, \mathbf{S}) \sim \xi}\left( \mathbf{N} ; \hat{\Pi}^{i}(\mathbf{U}) \given \mathbf{U}^i, \mathbf{S}\right)
		\\
		&\leq
		\MI_{
		(\mathbf{X}, \mathbf{M}, \mathbf{Z}) \sim \eta}
		\left( \mathbf{M}_{\mathbf{j}} ;  \Pi^{i}( \mathbf{X}) \given
		\mathbf{j},\mathbf{M}_{[1,\mathbf{j}-1]}, \mathbf{X}^i, \mathbf{Z}\right)
		\\
		&
		=
		\frac{1}{n}
		\sum_{j = 1}^n
		\MI_{(\mathbf{X}, \mathbf{M}, \mathbf{Z}) \sim \eta}
		\left( \mathbf{M}_j ;  \Pi^{i}( \mathbf{X}) \given
		\mathbf{M}_{[1,j-1]}, \mathbf{X}^i, \mathbf{Z}\right)
		=
		\frac{1}{n}
		\MI_{(\mathbf{X}, \mathbf{M}, \mathbf{Z}) \sim \eta} ( \mathbf{M} ;  \Pi^{i}( \mathbf{X}) \given \mathbf{X}^i, \mathbf{Z}).
	\end{align*}
The last step uses the chain rule.

Now let us prove the second inequality, which is quite similar.
We begin as before: by Lemma~\ref{lemma:chain}, since the tuple
$\langle \mathbf{j}, \mathbf{Z}_{\mathbf{-j}}, \mathbf{M}_{[1, \mathbf{j}-1]}, \mathbf{X}^i_{[\mathbf{j}+1,n]} \rangle$
is independent from $\mathbf{U}^i$ conditioned on $\mathbf{N}$ and $\mathbf{S}$,
	\begin{align}
		&\MI_{(\mathbf{U}, \mathbf{N}, \mathbf{S}) \sim \xi}\left( \mathbf{U}^i ; \hat{\Pi}^{i}(\mathbf{U}) \given \mathbf{N}, \mathbf{S}\right)
		\nonumber
		\\
		&=
		\MI_{
		\substack{(\mathbf{U}, \mathbf{N}, \mathbf{S}) \sim \xi\\
		(\mathbf{X}_{-\mathbf{j}}, \mathbf{M}_{-\mathbf{j}}, \mathbf{Z}_{-\mathbf{j}}) \sim \xi^{n-1}
		}}
		\left( \mathbf{U}^i ; \mathbf{j}, \mathbf{Z}_{-\mathbf{j}},\mathbf{M}_{[1,\mathbf{j}-1]}, \mathbf{X}^i_{[\mathbf{j}+1,n]}, \Pi^{i}( \embed(\mathbf{X}_{-\mathbf{j}}, \mathbf{j}, \mathbf{U}) ) \given \mathbf{N}, \mathbf{S}\right)
		\nonumber
		\\
		&=
		\MI_{
		\substack{(\mathbf{U}, \mathbf{N}, \mathbf{S}) \sim \xi\\
		(\mathbf{X}_{-\mathbf{j}}, \mathbf{M}_{-\mathbf{j}}, \mathbf{Z}_{-\mathbf{j}}) \sim \xi^{n-1}
		}}
		\left( \mathbf{U}^i ;  \Pi^{i}( \embed(\mathbf{X}_{-\mathbf{j}}, \mathbf{j}, \mathbf{U})) \given
		\mathbf{j},\mathbf{M}_{[1,\mathbf{j}-1]}, \mathbf{N},  \mathbf{X}^i_{[\mathbf{j}+1,n]}, \mathbf{Z}_{-\mathbf{j}}, \mathbf{S}\right)
		\nonumber
		\\
		&=
		\MI_{(\mathbf{X}, \mathbf{M}, \mathbf{Z}) \sim \eta}
		\left( \mathbf{X}_{\mathbf{j}}^i ;  \Pi^{i}( \mathbf{X}) \given
		\mathbf{j},\mathbf{M}_{[1,\mathbf{j}]}, \mathbf{X}^i_{[\mathbf{j}+1,n]}, \mathbf{Z}\right).
		\label{eq:Pihat1}
	\end{align}
	Next, since $\mathbf{M}_{[\mathbf{j}+1, n]}$ and $\mathbf{X}_{\mathbf{j}}^i$ (previously called $\mathbf{U}^i$) are independent
	given the conditioning in~\eqref{eq:Pihat1},
	we can apply Lemma~\ref{lemma:drop} to add conditioning on 
	$\mathbf{M}_{[\mathbf{j}+1,n]}$,
	yielding
	\begin{align*}
		&\MI_{(\mathbf{U}, \mathbf{N}, \mathbf{S}) \sim \xi}\left( \mathbf{N} ; \hat{\Pi}^{i}(\mathbf{U}) \given \mathbf{U}^i \mathbf{S}\right)
		\\
		&\leq
		\MI_{(\mathbf{X}, \mathbf{M}, \mathbf{Z}) \sim \eta}
		\left( \mathbf{X}_{\mathbf{j}}^i ;  \Pi^{i}( \mathbf{X}) \given
		\mathbf{j},\mathbf{M}, \mathbf{X}^i_{[\mathbf{j}+1,n]}, \mathbf{Z}\right)
		\\
		&
		=
		\frac{1}{n}
		\sum_{j = 1}^n
		\MI_{(\mathbf{X}, \mathbf{M}, \mathbf{Z}) \sim \eta}
		\left( \mathbf{X}_j^i ;  \Pi^{i}( \mathbf{X}) \given
		\mathbf{M}, \mathbf{X}_{[j+1,n]}^i, \mathbf{Z}\right)
		=
		\frac{1}{n}
		\MI_{(\mathbf{X}, \mathbf{M}, \mathbf{Z}) \sim \eta} ( \mathbf{X}^i ;  \Pi^{i}( \mathbf{X}) \given \mathbf{M}, \mathbf{Z}).
	\end{align*}
\end{proof}

The direct sum theorem follows immediately from Lemma~\ref{lem:directsum}:
\begin{theorem}
	Let $\xi$ be an $\eps$-collapsing distribution switched by $\mathbf{M}$ and $\mathbf{Z}$,
	where $\eps < (1 - \delta)/n$, and let $\eta = \xi^n$.
	Then
	\begin{equation*}
		\SIC_{\eta, \delta}(\disj_{n,k}) \geq n \cdot \SIC_{\xi,\delta+n \eps}(\AND_k).
	\end{equation*}
	\label{thm:directsum}
\end{theorem}

%% file: onebitsketch.tex
\section{The Information Complexity of One-Bit AND}
\label{sec:onebit}
By Theorem~\ref{thm:directsum}, in order to obtain an $\Omega(nk)$ lower bound on $\disj_{n,k}$ it is sufficient to show 
a lower bound of $\Omega(k)$ on the information complexity of $\AND_k$ under a hard one-bit distribution $\xi$,
which is both switched and $\eps$-collapsing. We will use the following distribution on $(\mathbf{X}, \mathbf{M}, \mathbf{Z})$ (informally described in Section~\ref{sec:overview}): \begin{itemize}
\item First we select $\mathbf{Z} \in_{\mathsf{U}} [k]$ and, independently, the mode $\mathbf{M}$
is selected with $\Pr[\mathbf{M} = 0] = 2/3$ and $\Pr[\mathbf{M} = 1] = 1/3$.
\item 
If $\mathbf{M} = 0$, then each player's input $\mathbf{X}^i$ is 0 or 1 with equal probability, independent of the other inputs.
If $\mathbf{M} = 1$, then the joint input is $\bar{e}_{\mathbf{Z}} \coloneq 1^{\mathbf{Z}-1}01^{k-\mathbf{Z}}$.
\end{itemize}

The distribution is switched by $\mathbf{M}$ and $\mathbf{Z}$, and is $\eps$-collapsing with $\eps = 1/(3 \cdot 2^{k-1})$.

\paragraph{Notation.}
In this section we let $\Pi(X)$ denote the distribution of the protocol's transcript when executed on input $X \in \set{0,1}^k$,
and similarly, $\Pi^i(X)$ denotes the distribution of player $i$'s view of the transcript.
We also let $\Pi^i[x,m,z]$ denote the distribution of player $i$'s view
when the input is drawn from $\xi$,
conditioned on $\mathbf{X}^i = x, \mathbf{M} = m$ and $\mathbf{Z} = z$.
For example, if $j \neq i$, then $\Pi^i[1,1,j] = \Pi(\bar{e}_j)$.
Notice that $\Pi^i[0,1,j]$ for $i \neq j$ is not well-defined, because $\Pr\left[ \mathbf{X}^i = 0, \mathbf{M} = 1, \mathbf{Z} \neq i \right] = 0$.
Similarly, we let $\Pi[i,x,m,z]$ denote the distribution of $\Pi$'s transcript,
conditioned on $\mathbf{X}_i = x, \mathbf{M} = m$ and $\mathbf{Z} = z$.
Finally, given a sequence $i_1, \ldots, i_{\ell} \in [k]$,
we use $\bar{e}_{i_1,\ldots,i_{\ell}}$ to denote the input in which players $i_1,\ldots,i_{\ell}$ receive zero,
and all other players receive one.

\subsection{Structural Properties of Protocols in the Coordinator Model}
We prove that $\SIC_{\xi,\delta}(\AND_k) = \Omega(k)$ in several steps.
The distribution $\xi$ comes in only when we relate Hellinger distance to mutual information;
for the most part we rely on the fact that $\Pi$ has error at most $\delta$ on any input,
and on the structural properties of $\Pi$.
We begin by outlining these properties.

The basic structural property on which we rely is \emph{rectangularity}, introduced in~\cite{baryossef04} for the two-player
setting and the multi-player model with communication by shared blackboard.
Rectangularity asserts, informally speaking, that if we partition the players into sets $A_1,\ldots,A_m \subseteq [k]$,
the protocol's probability distribution over transcripts
can be decomposed into a product of functions $f_1,\ldots,f_m$, such that each $f_i$ depends only in the inputs to players in $A_i$.
Here we require only a simple version where we use two sets, $A_1 = \set{i}$ and $A_2 = [k] \setminus \set{i}$ for some player $i \in [k]$.
The lemma follows by reduction from two-player rectangularity~\cite{baryossef04}, but
for the sake of completeness we include a proof.

\begin{lemma}[One-player rectangularity for the coordinator model]
	Let $\Pi$ be a 
$k$-player private-coin protocol
in the coordinator model, with
 inputs from $\mathcal{X} = \play{\mathcal{X}}{1} \times \ldots \times \play{\mathcal{X}}{k}$.
For $i \in [k]$, let $\play{\mathcal{T}}{i}$ denote the set of possible transcripts 
	observed by player $i$,
 so any transcript of $\Pi$ is in $ \play{\mathcal{T}}{1} \times \cdots \times \play{\mathcal{T}}{k}$.
	Then, for all $i \in [k]$, there exist mappings $\play{q}{i} : \play{\mathcal{X}}{i} \times \play{\mathcal{T}}{i} \rightarrow [0,1]$, $\play{q}{-i} : \play{\mathcal{X}}{-i} \times \play{\mathcal{T}}{i} \rightarrow [0,1]$
	and
$\play{p}{-i} : \play{\mathcal{X}}{-i} \times \mathcal{T} \rightarrow [0,1]$
	such that for any input $X \in \mathcal{X}$ and any transcript
$\tau = (\play{\tau}{1},\ldots,\play{\tau}{k}) \in \play{\mathcal{T}}{1} \times \cdots \times \play{\mathcal{T}}{k}$,
	\begin{align*}
		\Pr\left[ \play{\Pi}{i}(X) = \play{\tau}{i} \right] & = \play{q}{i}(\play{X}{i}, \play{\tau}{i}) \cdot \play{q}{-i}(\play{X}{-i}, \play{\tau}{i}) \mbox{ and }\\
		\Pr\left[ \Pi(X) = \tau \right] & = \play{q}{i}(\play{X}{i}, \play{\tau}{i}) \cdot \play{p}{-i}(\play{X}{-i}, \tau).
	\end{align*}
	\label{lemma:rect}
\end{lemma}
\begin{proof}
For any player $i$ and any transcript $\play{\tau}{i} \in \play{\mathcal{T}}{i}$,
let $\mathcal{A}(\play{\tau}{i})  = \{(X,R)\ |\ \plra{\Pi}{i}{R}(X) = \play{\tau}{i} \}$ denote
the set of inputs
and random coin tosses such that $\play{\tau}{i}$ is the transcript
$\plra{\Pi}{i}{R}(X)$
of the communication between player
$i$ and the coordinator
in the deterministic protocol $\rand{\Pi}{R}$ obtained from $\Pi$ 
by fixing the outcome of the random coin tosses to $R$.
Also, let
%For any player $i$, let
$\play{\mathcal{A}}{i}(\play{\tau}{i}) = \{ (\play{X}{i},\play{R}{i})\ | (X,R) \in \mathcal{A}(\play{\tau}{i})\}$
and 
$\play{\mathcal{A}}{-i}(\play{\tau}{i}) = \{ (\play{X}{-i},\play{R}{-i})\ | (X,R) \in \mathcal{A}(\play{\tau}{i})\}$.
Then, by the rectangular property for deterministic 2-player protocols,
%between player $i$ and the coordinator,
for all $(X,R)$, $\plra{\Pi}{i}{R}(X) = \play{\tau}{i}$ if and only if
$(\play{X}{i},\play{R}{i})\in \play{\mathcal{A}}{i}(\play{\tau}{i})$
and $(\play{X}{-i},\play{R}{-i})\in \play{\mathcal{A}}{-i}(\play{\tau}{i})$.

For any $\play{X}{i} \in \play{\mathcal{X}}{i}$, any
$\play{X}{-i} \in \play{\mathcal{X}}{-i}$, and any $\play{\tau}{i} \in \play{\mathcal{T}}{i}$,
define 
\begin{align*}
\play{q}{i} (\play{X}{i},\play{\tau}{i} ) & =
\Pr_{\play{R}{i}} \left [ (\play{X}{i},\play{R}{i}) \in \play{\mathcal{A}}{i}(\play{\tau}{i})\right ]
\mbox{ and}\\
%\Pr \left [ (Y^i,\play{R}{i}) \in \play{\mathcal{A}}{i}(\play{\tau}{i}) \ |\ Y^i = \play{X}{i} \right ]
\play{q}{-i} (\play{X}{-i},\play{\tau}{i} ) & =
\Pr_{\play{R}{-i}} \left [ (\play{X}{-i},\play{R}{-i}) \in \play{\mathcal{A}}{-i}(\play{\tau}{i}) \right ].
%\Pr \left [ (Y^{-i},\play{R}{-i}) \in \play{\mathcal{A}}{-i}(\play{\tau}{i}) \ |\ Y^{-i} = \play{X}{-i} \right ].
\end{align*}
On any input $X$, player $i$ chooses $\play{R}{i}$ uniformly
and the other players choose $\play{R}{-i}$ independently and uniformly.
Therefore,
\begin{align*} 
\Pr & \left[ \play{\Pi}{i}(X) = \play{\tau}{i} \right ]\\
& = \Pr_R \left[ \plra{\Pi}{i}{R}(X) = \play{\tau}{i} \right ]\\
&= \Pr_{\play{R}{i}} \left[ (\play{X}{i},\play{R}{i})\in \play{\mathcal{A}}{i}(\play{\tau}{i}) \right ] \cdot
\Pr_{\play{R}{-i}} \left[ (\play{X}{-i},\play{R}{-i}) \in \play{\mathcal{A}}{-i} (\play{\tau}{i}) \right ]\\
&= \play{q}{i} (\play{X}{i}, \play{\tau}{i} ) \cdot \play{q}{-i} ( \play{X}{-i}, \play{\tau}{i} ).
\end{align*}

For any transcript $\tau \in \mathcal{T}$,
let $\mathcal{B}(\tau)  = \{(X,R)\ |\ \rand{\Pi}{R}(X) = \tau \}$ denote
the set of inputs
and random coin tosses such that $\tau$ is the transcript $\rand{\Pi}{R}(X)$
of all communication (to and from the coordinator)
in the deterministic protocol $\rand{\Pi}{R}$ obtained from $\Pi$ 
by fixing the outcome of the random coin tosses to $R$.
Let
$\play{\mathcal{B}}{-i}(\tau) = \{ (\play{X}{-i},\play{R}{-i})\ | (X,R) \in \mathcal{B}(\tau)\}$.
By the rectangular property for deterministic protocols,
%between player $i$ and the coordinator,
for all $(X,R)$, $\rand{\Pi}{R}(X) = \tau$ if and only if
$(\play{X}{i},\play{R}{i})\in \play{\mathcal{A}}{i}(\play{\tau}{i})$
and $(\play{X}{-i},\play{R}{-i})\in \play{\mathcal{B}}{-i}(\tau)$.

For any $\play{X}{-i} \in \play{\mathcal{X}}{-i}$ and any $\tau \in \mathcal{T}$,
define
$$\play{p}{-i} (\play{X}{-i},\tau ) =
\Pr_{\play{R}{-i}} \left [ (\play{X}{-i},\play{R}{-i}) \in \play{\mathcal{B}}{-i}(\tau) \right ].$$
On any input $X$, player $i$ chooses $\play{R}{i}$ uniformly
and the other players choose $\play{R}{-i}$ independently and uniformly.
Therefore,
\begin{align*} 
\Pr & \left[ \Pi(X) = \tau \right ]\\
& = \Pr_R \left[ \rand{\Pi}{R}(X) = \tau \right ]\\
&= \Pr_{\play{R}{i}} \left[ (\play{X}{i},\play{R}{i})\in \play{\mathcal{A}}{i}(\play{\tau}{i}) \right ] \cdot
\Pr_{\play{R}{-i}} \left[ (\play{X}{-i},\play{R}{-i})\in \play{\mathcal{B}}{-i}(\tau) \right ]\\
&= \play{q}{i}(\play{X}{i}, \play{\tau}{i}) \cdot \play{p}{-i}(\play{X}{-i}, \tau).
\end{align*}
\end{proof}

For convenience,
when we apply Lemma~\ref{lemma:rect},
we sometimes write
$$\Pr\left[ \Pi(X) = \tau \right] = \play{p}{i}(\play{X}{i}, \tau) \cdot \play{p}{-i}(\play{X}{-i},\tau),$$
where  $\play{p}{i}(\play{X}{i}, \tau) = \play{q}{i}(\play{X}{i},\play{\tau}{i})$.

Rectangularity, in turn, implies the Z-Lemma (or Pythagorean Lemma) of~\cite{baryossef04}.
Here we use a simplified version (which omits one of the terms on the right-hand side):
\begin{lemma}[Diagonal Lemma]
	For any $X, Y \in \mathcal{X}$ and $\ell \in [k]$ we have
	\begin{equation*}
		h^2( \Pi(X), \Pi(Y)))
		\geq
		\frac{1}{2}
		h^2( \Pi(X), \Pi(\embed(Y_{-\ell}, \ell, X_\ell))).
	\end{equation*}
	\label{lemma:Z}
\end{lemma}

Under our distribution $\xi$, the inputs $\mathbf{X}^i$ are independent given $\mathbf{M}$ and $\mathbf{Z}$.
This allows us to prove the following variant of the rectangular property, which, informally speaking,
``abstracts away'' all the inputs $\mathbf{X}^{-i}$ by grouping them together under the conditioning $\mathbf{M} = m, \mathbf{Z} = z$
(for some $m$ and $z$).
\begin{lemma}[Conditional rectangularity for $\mathbf{M}$ and $\mathbf{X}^i$ under $\xi$]
	Let $\Pi$ be a $k$-player private-coin protocol for $\AND_k$.
For $i \in [k]$, 
let $\play{\mathcal{T}}{i}$ denote the set of possible transcripts 
	observed by player $i$.
	Then
there exists a function 
$c : \{0,1\} \times [k] \times \mathcal{T} \rightarrow [0,1]$
and, for all $i \in [k]$,
there exists a function
$\play{c}{i} : \{0,1\} \times [k] \times \play{\mathcal{T}}{i} \rightarrow [0,1]$
	such that for any
	$x \in \play{\mathcal{X}}{i}$,
	$m \in \{0,1\}$,
	$z \in [k] \setminus \set{i}$, $\tau \in \mathcal{T}$, and $\play{\tau}{i} \in \play{\mathcal{T}}{i}$,
	\begin{align*}
\Pr [ 
\Pi(\mathbf{X})
&= \tau \given \mathbf{X}^i = x, \mathbf{M} = m, \mathbf{Z} = z  ]
= \play{p}{i}(x, \tau) \cdot c(d, z, \tau) \mbox{ and}
\\
\Pr [ \play{\Pi}{i}(\mathbf{X}) &= \play{\tau}{i} \given \mathbf{X}^i = x, \mathbf{M} = m, \mathbf{Z} = z  ] = \play{q}{i}(x, \play{\tau}{i}) \cdot \play{c}{i}(d, z, \play{\tau}{i}),
	\end{align*}
	where $\play{p}{i}(x, \tau) = \play{q}{i}(x, \play{\tau}{i})$ is the 
function from Lemma~\ref{lemma:rect}.
	Here the probability is over the protocol's own randomness as well as the input $\mathbf{X}$ drawn from $\xi$
with the stated conditioning.
	\label{lemma:rectD}
\end{lemma}
\begin{proof}
	By Lemma~\ref{lemma:rect} there exist functions $\play{q}{i}$ and
$\play{q}{-i}$ such that,	
	for any input $X \in \mathcal{X}$,
	\begin{equation*}
		\Pr\left[ \play{\Pi}{i}(X) = \play{\tau}{i} \right] = \play{q}{i}(\play{X}{i}, \play{\tau}{i}) \cdot \play{q}{-i}(\play{X}{-i}, \play{\tau}{i}).
	\end{equation*}
	Therefore we can write 
	\begin{align*}
		&\Pr\left[ \play{\Pi}{i}(\mathbf{X}) = \play{\tau}{i} \given \mathbf{M} = m, \mathbf{Z} = z, \play{\mathbf{X}}{i} = x  \right]
		\\
		&= \sum_{X \in \mathcal{X}} \play{q}{i}(\play{X}{i}, \play{\tau}{i}) \cdot \play{q}{-i}(\play{X}{-i}, \play{\tau}{i}) \cdot \Pr\left[ \mathbf{X} = X \given \mathbf{M} = m, \mathbf{Z} = z, \play{\mathbf{X}}{i} = x \right]
		\\
%&= \sum_{X_{-i} \in \mathcal{X}_{-i}}
%\play{p}{i}(x, \tau) \cdot \play{p}{-i}(\play{X}{-i}, \tau) \cdot 
%\Pr\left[ \play{\mathbf{X}}{i} = \play{X}{i} \given \play{\mathbf{X}}{i} = x, \mathbf{D} = d, \mathbf{Z} = z\right]
%\Pr\left[ \mathbf{X}^{-i} = \play{X}{-i} \given \mathbf{D} = d, \mathbf{Z} = z %\right]
%		\\
		&=
		\play{q}{i}(x, \play{\tau}{i}) \cdot \sum_{\play{X}{-i} \in \play{\mathcal{X}}{-i}} \play{q}{-i}(\play{X}{-i}, \play{\tau}{i}) \cdot \Pr\left[ \play{\mathbf{X}}{-i} = \play{X}{-i} \given \mathbf{M} = m, \mathbf{Z} = z \right].
	\end{align*}
	Here we use the fact that the inputs $\play{\mathbf{X}}{1}, \ldots,
\play{\mathbf{X}}{k}$ are independent conditioned on $\mathbf{M}$ and $\mathbf{Z}$.
	The second claim follows by setting
	\begin{equation*}
		\play{c}{i}(d, z, \play{\tau}{i}) = \sum_{\play{X}{-i} \in \play{\mathcal{X}}{-i}} \play{q}{-i}(\play{X}{-i}, \play{\tau}{i}) \cdot \Pr\left[ \play{\mathbf{X}}{-i} = \play{X}{-i} \given \mathbf{<} = m, \mathbf{Z} = z \right].
	\end{equation*}

Similarly, Lemma~\ref{lemma:rect} implies there exist functions $\play{p}{i}$ and
$\play{p}{-i}$ such that,	
	for any input $X \in \mathcal{X}$,
	\begin{equation*}
		\Pr\left[ \Pi(X) = \tau \right] = \play{p}{i}(\play{X}{i}, \tau) \cdot \play{p}{-i}(\play{X}{-i}, \tau), 
	\end{equation*}
so
\begin{align*}
		&\Pr\left[ \Pi(\mathbf{X}) = \tau \given \mathbf{M} = m, \mathbf{Z} = z, \play{\mathbf{X}}{i} = x  \right]
		\\
&=
		\play{p}{i}(x, \tau) \cdot \sum_{\play{X}{-i} \in \play{\mathcal{X}}{-i}} \play{p}{-i}(\play{X}{-i}, \tau) \cdot \Pr\left[ \play{\mathbf{X}}{-i} = \play{X}{-i} \given \mathbf{M} = m, \mathbf{Z} = z \right]
\end{align*}
and the first claim follows by setting
	\begin{equation*}
		c(d, z, \tau) = \sum_{\play{X}{-i} \in \play{\mathcal{X}}{-i}} \play{p}{-i}(\play{X}{-i}, \tau) \cdot \Pr\left[ \play{\mathbf{X}}{-i} = \play{X}{-i} \given \mathbf{M} = m, \mathbf{Z} = z \right].
	\end{equation*}
\end{proof}

Lemma~\ref{lemma:rectD} yields the following variant of the Diagonal Lemma (Lemma~\ref{lemma:Z}).
\begin{lemma}[Diagonal Lemma for $\mathbf{M}$ and $\mathbf{X}^i$]
	For any $i \neq z$ we have
	\begin{equation*}
		h^2(\Pi^i[0,0,z], \Pi^i[1,1,z])
		\geq
		\frac{1}{2}
		h^2( \Pi^i(\bar{e}_{i,z}), \Pi^i(\bar{e}_z) ).
	\end{equation*}
	\label{lemma:ZD}
\end{lemma}
\begin{proof}
	The proof closely follows  the proof of the original Z-Lemma from~\cite{baryossef04},
	but we include it here for completeness.

	Recall that $\Pi^i[1,1,z] = \Pi^i(\bar{e}_z)$.
	By Lemmas~\ref{lemma:rect} and~\ref{lemma:rectD}, we can decompose
	the distributions from the lemma statement as follows:
	\begin{align*}
&
\Pr\left[ \play{\Pi}{i}[0,0,z] = \tau^i \right]
		= 
		\play{q}{i}(0, \play{\tau}{i}) \cdot \play{c}{i}( 0, z, \play{\tau}{i}),
		\\
		&
		\Pr\left[ \Pi^i[1,1,z] = \tau^i \right] = \Pr\left[ \Pi^i(\bar{e}_z) = \tau^i \right] = 
		\play{q}{i}( 1, \play{\tau}{i}) \cdot \play{q}{-i}( (\bar{e}_z)^{-i}, \play{\tau}{i})
		= \play{q}{i}(1, \play{\tau}{i}) \cdot \play{c}{i}( 1, z, \play{\tau}{i}),
		\medspace \text{ and}\\
		&
		\Pr\left[ \Pi^i(\bar{e}_{i,z}) = \tau^i \right] = \play{q}{i}( 0, \tau^i) \cdot \play{q}{-i}( (\bar{e}_{i,z})^{-i}, \tau^i)
		= \play{q}{i}( 0, \tau^i) \cdot \play{q}{-i}( (\bar{e}_z)^{-i}, \tau^i).
	\end{align*}

%	Here $\play{q}{-i}( \play{(\chi_{z})}{-i},\play{\tau}{i}) 
%= \play{c}{i}( 1, z, \play{\tau}{i})$ and we will use this fact below.
From the definition of Hellinger distance, it follows that
	\begin{align*}
1 - h^2( 
\Pi^i[0,0,z], \Pi^i[1,1,z])
		&=
		\sum_{\play{\tau}{i}} \sqrt{ 
\play{q}{i}(0, \play{\tau}{i}) \cdot \play{c}{i}( 0, z, \play{\tau}{i})
\cdot \play{q}{i}(1, \play{\tau}{i}) \cdot \play{c}{i}( 1, z, \play{\tau}{i})}
		\\
		&\leq
		\sum_{\play{\tau}{i}} \sqrt{
 \play{q}{i}(0, \play{\tau}{i})\play{q}{i}(1, \play{\tau}{i})}
\left ( \frac{\play{c}{i}( 0, z, \play{\tau}{i}) + \play{c}{i}( 1, z, \play{\tau}{i})}{2}
\right )
		\\
		&=
		\frac{1}{2} \sum_{\play{\tau}{i}} \sqrt{
\play{q}{i}(0, \play{\tau}{i})\play{c}{i}( 0, z, \play{\tau}{i})
\play{q}{i}(1, \play{\tau}{i})\play{c}{i}( 0, z, \play{\tau}{i})} \\
 &\qquad + \frac{1}{2} \sum_{\play{\tau}{i}}\sqrt{\play{q}{i}(0, \play{\tau}{i})
 \play{q}{-i}( (\bar{e}_z){-i},\play{\tau}{i}) 
\play{q}{i}(1, \play{\tau}{i})
\play{q}{-i}( (\bar{e}_z)^{-i},\play{\tau}{i}) }
		\\
		&=
		\left( 1 -
		h^2( \Pi^i[0,0,z], \Pi^i[1,0,z])\right) / 2
		+ 
		\left( 1 -
		h^2( \Pi^i(\bar{e}_{i,z}), \Pi^i(\bar{e}_z)\right) / 2
		\\
		&
		=
		1 - 
		\left(h^2( \Pi^i[0,0,z], \Pi^i[1,0,z])
		+
		h^2( \Pi^i(\bar{e}_{i,z}), \Pi^i(\bar{e}_z)\right) / 2
		\\
		&\leq
		1 - 
		h^2( \Pi^i(\bar{e}_{i,z}), \Pi^i(\bar{e}_z)) / 2.
	\end{align*}
\end{proof}

Note that Lemma~\ref{lemma:Z} concerns the complete transcript $\Pi$,
while Lemma~\ref{lemma:ZD} concerns one player's local view, $\Pi^i$.
To move between the two we
 use the following ``localization'' lemma,
 which shows that when we ``keep everything the same'' and change only $\mathbf{X}^i$,
	the distance between the transcript's distributions is caused entirely by player $i$'s
	local view.
	What does it mean to ``keep everything the same except $\mathbf{X}^i$''?
	One option is to fix $\mathbf{M} = m$ and a specific value $\mathbf{Z} = z \neq i$,
	and let $\mathbf{X}^i$ change from 0 to 1.
	This is well-defined only in the case where $m = 0$, because when $\mathbf{M} = 1$ and $\mathbf{Z} \neq i$,
	we must have $\mathbf{X}^i = 1$.
	The other option is to 
	fix a specific input $\mathbf{X}^{-i} = X^{-i}$ for the rest of the players
	and let $\mathbf{X}^i$ change from 0 to 1.
	We are particularly interested in the case where all players receive 1, except for one player, $z \neq i$,
	and possibly player $i$ itself.
\begin{lemma}[Localizing the distance to a single player's transcript]
	For any $i \neq z$ we have
	\begin{align*}
		& h( \Pi[i, 0, 0, z], \Pi[i, 1, 0, z] )
		=
		h( \Pi^i[0, 0, z], \Pi^i[1, 0, z] ),
	\end{align*}
	and similarly,
	\begin{equation*}
		h( \Pi(\bar{e}_{i,z}), \Pi(\bar{e}_z)) = h( \Pi^i(\bar{e}_{i,z}), \Pi^i(\bar{e}_z)).
	\end{equation*}
\end{lemma}
\begin{proof}
	Given a complete transcript $\tau$, let $\tau^i$ denote player $i$'s part of the transcript.
	By Lemma~\ref{lemma:rectD},
\begin{align*}
\Pr [ \Pi(\mathbf{X}) &= \tau \given \play{\mathbf{X}}{i} = x, \mathbf{M} = 0, \mathbf{Z} = z ] = \play{p}{i}(x, \tau) \cdot c(0, z, \tau) \mbox{ and}\\
\Pr [\play{\Pi}{i}(\mathbf{X}) &= \play{\tau}{i} \given \play{\mathbf{X}}{i} = x, \mathbf{M} = 0, \mathbf{Z} = z] = \play{q}{i}(x,  \play{\tau}{i}) \cdot \play{c}{i}(0, z,  \play{\tau}{i}).
\end{align*}
	Moreover, 
	\begin{equation*}
		\Pr [\play{\Pi}{i}(\mathbf{X}) = \play{\tau}{i} \given \play{\mathbf{X}}{i} = x, \mathbf{M} = 0, \mathbf{Z} = z] =
		\sum_{\begin{array}[t]{c}{\scriptstyle \tau' \in \mathcal{T}}\\
[-6pt] {\scriptstyle \play{\tau'}{i} = \play{\tau}{i} } \end{array}}
		\Pr[\Pi(\mathbf{X}) = \tau' \given \play{\mathbf{X}}{i} = x, \mathbf{M} = 0, \mathbf{Z} = z],
	\end{equation*}
so,
	\begin{equation*}
\play{c}{i}(d, z, \play{\tau}{i}) =
\sum_{\begin{array}[t]{c}{\scriptstyle \tau' \in \mathcal{T}}\\
[-6pt] {\scriptstyle \play{\tau'}{i} = \play{\tau}{i} } \end{array}}
%\sum_{\tau' \in \mathcal{T}\\ \play{\tau'}{i} = \play{\tau}{i} }
c(d, z, \tau').
	\end{equation*}
	Therefore,
	\begin{align*}
		&1 - h^2( \play{\Pi}{i}[0, 0, z], \play{\Pi}{i}[1,0,z]) 
		\\
		&=
		\sum_{\play{\tau}{i}} \sqrt{ \play{q}{i}(0, \play{\tau}{i}) 
\play{c}{i}(d, z, \play{\tau}{i}) \play{q}{i}(1, \play{\tau}{i}) \play{c}{i}(d, z, \play{\tau}{i})}
		\\
		&=
		\sum_{ \play{\tau}{i}} \left( \sqrt{ \play{q}{i}(0, \play{\tau}{i})\play{q}{i}(1, \play{\tau}{i}) } \right .
\sum_{\begin{array}[t]{c}{\scriptstyle \tau' \in \mathcal{T}}\\
[-6pt] {\scriptstyle \play{\tau'}{i} = \play{\tau}{i} } \end{array}}
\left . c(d, z, \tau') \rule{0mm}{4mm} \right )
		\\
		&=
		\sum_{\tau } \sqrt{ \play{q}{i}(0, \play{\tau}{i})\play{q}{i}(1, \play{\tau}{i})}
		c(d, z, \tau)\\
		&=
		\sum_{\tau }
\sqrt{ 
\play{q}{i}(0, \play{\tau}{i})c(d, z, \tau)\play{q}{i}(1, \play{\tau}{i})c(d, z, \tau)}
		\\
		&
		= 1 - 
		h^2( \Pi[i, 0, 0, z], \Pi[i, 1, 0, z]).
	\end{align*}

	The other part of the lemma is similar: it is obtained by using Lemma~\ref{lemma:rect}
	instead of Lemma~\ref{lemma:rectD} and replacing $c, c^i$ with $q^{-i}, p^{-i}$ (respectively).
\end{proof}

Now we are ready to describe the main proof that the information complexity of $\AND_k$ is $\Omega(k)$.

\subsection{Step I: setting up a rectangle.}
Fix a player $i$ and a value $z \neq i$, and consider the following four distributions:
\begin{center}
\begin{tabular}[h]{cc}
	$\Pi^i[0, 0, z]$ & $\Pi^i[1, 0, z]$ \\
	\\
	$\Pi^i(\bar{e}_{i,z})$ & $\Pi^i(\bar{e}_z) = \Pi^i[1, 1, z]$
\end{tabular}
\end{center}
The two distributions in the top row differ only in the value of $\mathbf{X}_i$,
which is 0 for the first column and 1 for the second;
the same holds for the bottom row.
The top-row distributions have $\mathbf{M} = 0$, and it is helpful to think of the bottom
row as representing the hard case, $\mathbf{M} = 1$ (although $\Pi^i[0,1,z]$ is not well-defined, and moreover, the input $\bar{e}_{i,z}$ has probability 0 under $\xi$).

Notice that our distribution $\xi$ has the following nice property:
\begin{align*}
	& \Pr\left[ \mathbf{X}^i = 0 \given \mathbf{M} = 0, \mathbf{Z} = z \right] = \Pr\left[ \mathbf{X}^i = 1 \given \mathbf{M} = 0, \mathbf{Z} = z \right] = 1/2,
	\quad \text{ and}\\
	&\Pr\left[ \mathbf{M} = 0 \given \mathbf{X}^i = 1, \mathbf{Z} = z \right] = \Pr\left[ \mathbf{M} = 1 \given \mathbf{X}^i = 1, \mathbf{Z} = z \right] = 1/2.
\end{align*}
In other words, given that we are in the top row of the rectangle ($\mathbf{M} = 0, \mathbf{Z} = z$),
the distribution of the transcript $\Pi^i$ is equally likely to be $\Pi^i[0,0,z]$ or $\Pi^i[1,0,z]$, 
the two top-row distributions.
This means that \emph{if the two top-row distributions have a large Hellinger distance}, then
the conditional mutual information
$\MI(\mathbf{X}_i ; \Pi^i \given \mathbf{M} = 0, \mathbf{Z} = z)$ is large:
although $\mathbf{X}_i$ is equally likely to be 0 or 1 \emph{a priori} given $\mathbf{M} = 0, \mathbf{Z} = z$,
because of the large Hellinger distance,
the transcript $\Pi^i$ allows us to distinguish the case $\mathbf{X}_i = 0$ from the case $\mathbf{X}_i = 1$.
This is captured by Lemma~\ref{lemma:MI}, which yields
\begin{align*}
	&\MI( \mathbf{M} ; \Pi^i \given \mathbf{X}^i = 1, \mathbf{Z} = z) \geq
	h(\Pi^i[1, 0, z], \Pi^i[1, 1, z]).
\end{align*}

Similarly, given that we are in the rightmost column ($\mathbf{X}_i = 1, \mathbf{Z} = z$),
the distribution of $\Pi^i$ is equally likely to be $\Pi^i[1,0,z]$ or $\Pi^i[1,1,z]$.
Therefore a large Hellinger distance between these distributions implies that $\MI( \mathbf{M} ; \Pi^i \given \mathbf{X}_i = 1, \mathbf{Z} = z)$ is large:
Lemma~\ref{lemma:MI} again yields
\begin{align*}
	\MI( \mathbf{X}^i ; \Pi^i \given \mathbf{M} = 0, \mathbf{Z} = z) \geq h(\Pi^i[0, 0, z], \Pi^i[1, 0, z]).
\end{align*}

Recall that
$\Pr\left[ \mathbf{M} = 0 \given \mathbf{Z} = z \right] = 2/3$ (as $\mathbf{M}$ and $\mathbf{Z}$ are independent),
and observe that
when $z \neq i$ we have $\Pr\left[ \mathbf{X}^i = 1 \given \mathbf{Z} = z \right] = 2/3$.
Therefore
$\MI\left( \mathbf{X}^i ; \Pi^i \given \mathbf{M}, \mathbf{Z} = z \right) \geq (2/3)\MI\left( \mathbf{X}^i ; \Pi^i \given \mathbf{M} = 0, \mathbf{Z} = z \right)$ and 
$\MI\left( \mathbf{M} ; \Pi^i \given \mathbf{X}^i, \mathbf{Z} = z \right) \geq (2/3)\MI\left( \mathbf{M} ; \Pi^i \given \mathbf{X}^i = 1, \mathbf{Z} = z \right)$.
It follows that
\begin{align*}
	\MI( \mathbf{M} ; \Pi^i \given \mathbf{X}^i, \mathbf{Z} = z)
	+
	\MI( \mathbf{X}^i ; \Pi^i \given \mathbf{M}, \mathbf{Z} = z) 
	&
	\geq
	\frac{2}{3}
	\left(
	h^2(\Pi^i[1, 0, z], \Pi^i[1, 1, z])
	+
	h^2(\Pi^i[0, 0, z], \Pi^i[1, 0, z])
	\right)
	\\
	&
	\geq
	\frac{
	\left(
	h(\Pi^i[1, 0, z], \Pi^i[1, 1, z])
	+
	h(\Pi^i[0, 0, z], \Pi^i[1, 0, z])
	\right)^2}{3}
	\\
	&\geq
	\frac{
	h^2(\Pi^i[0, 0, z], \Pi^i[1, 1, z])
	}{3}.
\end{align*}
The last step uses the triangle inequality.
Now we apply Lemma~\ref{lemma:ZD}, which together with the above yields
\begin{equation}
	\MI( \mathbf{M} ; \Pi^i \given \mathbf{X}^i, \mathbf{Z} = z)
	+
	\MI( \mathbf{X}^i ; \Pi^i \given \mathbf{M}, \mathbf{Z} = z) 
	\geq
	\frac{
	h^2( \Pi^i(\bar{e}_{i,z}), \Pi^i(\bar{e}_z) )
	}
	{3}.
	\label{eq:MI_h}
\end{equation}
This holds only for $z \neq i$.
Taking the expectation over \emph{all} $z \in [k]$, we obtain
\begin{align}
	&\MI( \mathbf{M} ; \Pi^i \given \mathbf{X}^i, \mathbf{Z})
	+
	\MI( \mathbf{X}^i ; \Pi^i \given \mathbf{M}, \mathbf{Z}) 
	\geq
	\frac{1}{k}
	\sum_{z \neq i}
	\left(
	\MI( \mathbf{M} ; \Pi^i \given \mathbf{X}^i, \mathbf{Z} = z)
	+
	\MI( \mathbf{X}^i ; \Pi^i \given \mathbf{M}, \mathbf{Z} = z) 
	\right)
	\nonumber
	\\
	&
	\stackrel{\eqref{eq:MI_h}}{\geq}
	\frac{k-1}{3k} \E_{\mathbf{Z} \neq i}\left[ h^2( \Pi^i(\bar{e}_{i,\mathbf{Z}}, \Pi^i(\bar{e}_{\mathbf{Z}})) \right]
	\geq
	\frac{1}{6} \E_{\mathbf{Z} \neq i}\left[ h^2( \Pi^i(\bar{e}_{i,\mathbf{Z}}, \Pi^i(\bar{e}_{\mathbf{Z}})) \right].
	\label{eq:MI_gamma}
\end{align}
The last step uses the fact that $k - 1 \geq k/2$, as $k > 1$.

Let us define the \emph{usefulness of player $i$} to be
$\gamma_i \coloneq \E_{\mathbf{Z} \neq i}\left[ h^2( \Pi^i(\bar{e}_{i,\mathbf{Z}}, \Pi^i(\bar{e}_{\mathbf{Z}})) \right]$.
Roughly speaking, player $i$'s usefulness corresponds to how sensitive the protocol is to the fact that $\mathbf{X}^i = 0$,
\emph{when some other player $z \neq i$ also has 0}.
By~\eqref{eq:MI_gamma} we see that in order to obtain our desired $\Omega(k)$ lower bound, it is sufficient to bound the sum $\sum_i \gamma_i$
(or the average, $\sum_i \gamma_i / k$).
But why should $\gamma_i$ be large on average? In other words, why should the protocol distinguish the case where only one player
has zero from the case where two players have zero, when the answer to $\AND_k$ is 0 in both cases?
This will again follow from the structural properties of the protocol.

\subsection{Step II: bounding the average usefulness.}
In order to show that the average player has a large usefulness $\gamma_i$,
consider any two players $i \neq j$,
and the following four distributions:
\begin{center}
\begin{tabular}[h]{ll}
	$\Pi(\bar{e}_i)$ & $\Pi(1^k)$ \\
	$\Pi(\bar{e}_{i,j})$ & $\Pi(\bar{e}_j)$
\end{tabular}
\end{center}
We have $\AND_k(\bar{e}_i) = \AND_k(\bar{e}_j) = 0$, but $\AND_k(1^k) = 1$.
By the correctness of the protocol and Lemma~\ref{lemma:h_delta}, the statistical distance between $\Pi(\bar{e}_i)$ and $\Pi(1^k)$ must be at least $1 - \delta$,
which implies that $h( \Pi(\bar{e}_i), \Pi(1^k) ) \geq (1 - \delta)/\sqrt{2}$.
By the diagonal lemma (with $\ell = j$),
$h( \Pi(\bar{e}_i), \Pi(\bar{e}_j) ) \geq h( \Pi( \bar{e}_i), \Pi(1^k) ) / \sqrt{2} \geq (1 - \delta)/2$,
that is, the protocol must distinguish $\bar{e}_i$ from $\bar{e}_j$.
(Roughly speaking, this means that the protocol must \emph{find} a player that has zero in the case where $\mathbf{M} = 1$,
an interesting fact in itself.)
By the triangle inequality,
\begin{equation*}
	h( \Pi( \bar{e}_i), \Pi(\bar{e}_{i,j})) + h(\Pi(\bar{e}_j), \Pi(\bar{e}_{i,j})) \geq h( \Pi(\bar{e}_i), \Pi(\bar{e}_j)) \geq (1 - \delta)/2,
\end{equation*}
and therefore
\begin{align*}
	h^2( \Pi( \bar{e}_i), \Pi(\bar{e}_{i,j})) + h^2(\Pi(\bar{e}_j), \Pi(\bar{e}_{i,j}))
	&\geq
	\frac{
	\left(h( \Pi( \bar{e}_i), \Pi(\bar{e}_{i,j})) + h(\Pi(\bar{e}_j), \Pi(\bar{e}_{i,j}))\right)^2}
	{2}
	\geq \frac{(1 - \delta)^2}{8}.
\end{align*}
Now summing across all players $i \neq j$, we see that $2\sum_i \sum_{j \neq i} h^2( \Pi( \bar{e}_i), \Pi(\bar{e}_{i,j})) \geq k(k-1) \cdot (1-\delta)^2/8$,
which implies that $\sum_i \gamma_i \geq k \cdot (1 - \delta)^2 / 16 = \Omega(k)$.
Together with~\eqref{eq:MI_gamma}, this yields our main result for this section:

\begin{theorem}
	For any $k > 1$, $\SIC_{\xi, \delta}(\AND_k) \geq (1 - \delta)^2 / 96$.
	\label{thm:onebit}
\end{theorem}

Combining Theorem~\ref{thm:onebit} with our direct-sum theorem from Section~\ref{sec:directsum},
we obtain 
\begin{theorem}\label{thm:sicdisj}
	For any $n \geq 1$ and for $k = \Omega(\log n)$,
	$\SIC_{\eta, \delta}(\disj_{n,k}) = \Omega(nk)$.
	\label{thm:disj}
\end{theorem}

%% file: marktemp.tex
\section{Internal Information Complexity and Communication Complexity of Set Disjointness}
\label{sec:mark}

Recall our definition of the internal information cost of a protocol from Section~\ref{sec:prelim}:
	\begin{equation*}
		IC_{\zeta}(\Pi)
		\coloneq
			\MI_{\mathbf{X} \sim \zeta}( \mathbf{X} ; \Pi(\mathbf{X}) ) + 
		\sum_{i \in [k]} \left[
			\MI_{\mathbf{X} \sim \zeta}( \mathbf{X}^{-i} ; \Pi^{i}(\mathbf{X}) \given \mathbf{X}^i)
		\right].
	\end{equation*}
	We will now use Theorem~\ref{thm:disj} to show that the internal information complexity of set disjointness is 
	$\Omega(nk)$. 
	
	In Theorem~\ref{thm:disj} we showed that for all protocols $\Pi$ we have
$$
	\SIC_{\eta}(\Pi)
		=
		\sum_{i \in [k]} \left[
		\MI_{(\mathbf{X}, \mathbf{M}, \mathbf{Z}) \sim \eta}( \mathbf{X}^i ; \Pi^{i}(\mathbf{X}) \given \mathbf{M}, \mathbf{Z})
		+
		\MI_{(\mathbf{X}, \mathbf{M}, \mathbf{Z}) \sim \eta}( \mathbf{M} ; \Pi^{i}(\mathbf{X}) \given \mathbf{X}^i, \mathbf{Z})
		\right] = \Omega(n\cdot k). 
$$
Let $\zeta$ be the distribution $\eta$ restricted just to the input $\mathbf{X}$ (that is, $\zeta$
is the marginal distribution of $\mathbf{X}$ under $\eta$).
We will show:
\begin{theorem}\label{lem:ICsetdisj}
Let $\Pi$ be a protocol  in the coordinator model.
Let $\eta$ be a switched distribution and $\zeta$ be 
the marginal distribution of $\mathbf{X}$ under $\eta$, as above. Namely, $\zeta$ is the marginal distribution of $\mathbf{X}$
where $(\mathbf{X},\mathbf{M}, \mathbf{Z}) \sim \eta$. 
Then, 
$$
IC_\zeta(\Pi) > \SIC_{\eta}(\Pi) - O(n \log k). %= \Omega(n\cdot k). 
$$
\end{theorem}
\begin{proof}
We consider the two terms in each sum separately. We start with the term corresponding to the amount of information learned 
by the coordinator. By the definition of $\eta$ we have that $I(\mathbf{X}^i; \mathbf{X}^{[1..i-1]}\given\mathbf{M},\mathbf{Z})=0$ and thus by 
Lemma~\ref{lemma:drop}, 
$$
I( \mathbf{X}^i ; \Pi^{i}(\mathbf{X}) \given \mathbf{M}, \mathbf{Z}) \le I( \mathbf{X}^i ; \Pi^{i}(\mathbf{X}) \given \mathbf{M}, \mathbf{Z},\mathbf{X}^{[1..i-1]}).
$$
Using the Chain Rule, we get:
\begin{eqnarray*}
	\sum_{i \in [k]} \left[
		\MI_{(\mathbf{X}, \mathbf{M}, \mathbf{Z}) \sim \eta}( \mathbf{X}^i ; \Pi^{i}(\mathbf{X}) \given \mathbf{M}, \mathbf{Z})\right] & \le &  
\sum_{i \in [k]} \left[
		\MI_{(\mathbf{X}, \mathbf{M}, \mathbf{Z}) \sim \eta}( \mathbf{X}^i ; \Pi^i(\mathbf{X}) \given \mathbf{M}, \mathbf{Z},\mathbf{X}^{[1..i-1]} )\right] \\
& \le & \sum_{i \in [k]} \left[
		\MI_{(\mathbf{X}, \mathbf{M}, \mathbf{Z}) \sim \eta}( \mathbf{X}^i ; \Pi(\mathbf{X}) \given \mathbf{M}, \mathbf{Z},\mathbf{X}^{[1..i-1]} )\right] \\ 
		&  = & 
			\MI_{(\mathbf{X}, \mathbf{M}, \mathbf{Z}) \sim \eta}( \mathbf{X} ; \Pi(\mathbf{X}) \given \mathbf{M}, \mathbf{Z} ) \\ & \le &
				\MI_{(\mathbf{X}, \mathbf{M}, \mathbf{Z}) \sim \eta}( \mathbf{X},\mathbf{M}, \mathbf{Z} ; \Pi(\mathbf{X})  ) \\ &  = & 
					\MI_{(\mathbf{X}, \mathbf{M}, \mathbf{Z}) \sim \eta}( \mathbf{X}; \Pi(\mathbf{X})  ) +
					\MI_{(\mathbf{X}, \mathbf{M}, \mathbf{Z}) \sim \eta}( \mathbf{M}, \mathbf{Z} ; \Pi(\mathbf{X}) \given \mathbf{X} ) \\ & \le & 
					\MI_{(\mathbf{X}, \mathbf{M}, \mathbf{Z}) \sim \eta}( \mathbf{X}; \Pi(\mathbf{X})  ) + H( \mathbf{M}, \mathbf{Z}) \\ & \le & 
					\MI_{(\mathbf{X}, \mathbf{M}, \mathbf{Z}) \sim \eta}( \mathbf{X}; \Pi(\mathbf{X})  ) + O(n\log k)
\end{eqnarray*}

Next, we consider the terms corresponding to what individual players learn. For each $i\in [k]$ we have 
\begin{eqnarray*}
	\MI_{(\mathbf{X}, \mathbf{M}, \mathbf{Z}) \sim \eta}( \mathbf{M} ; \Pi^{i}(\mathbf{X}) \given \mathbf{X}^i, \mathbf{Z}) & \le & 
		\MI_{(\mathbf{X}, \mathbf{M}, \mathbf{Z}) \sim \eta}( \mathbf{X}^{-i} ; \Pi^{i}(\mathbf{X}) \given \mathbf{X}^i, \mathbf{Z}) \\ & \le & 	
		\MI_{(\mathbf{X}, \mathbf{M}, \mathbf{Z}) \sim \eta}( \mathbf{X}^{-i}, \mathbf{Z} ; \Pi^{i}(\mathbf{X}) \given \mathbf{X}^i) \\ & = &
			\MI_{(\mathbf{X}, \mathbf{M}, \mathbf{Z}) \sim \eta}( \mathbf{X}^{-i} ; \Pi^{i}(\mathbf{X}) \given \mathbf{X}^i) + 
				\MI_{(\mathbf{X}, \mathbf{M}, \mathbf{Z}) \sim \eta}( \mathbf{Z} ; \Pi^{i}(\mathbf{X}) \given \mathbf{X}^i, \mathbf{X}^{-i}) \\ & \le & 
		\MI_{(\mathbf{X}, \mathbf{M}, \mathbf{Z}) \sim \eta}( \mathbf{X}^{-i} ; \Pi^{i}(\mathbf{X}) \given \mathbf{X}^i) + 
		H(\mathbf{Z}) \le 	\MI_{(\mathbf{X}, \mathbf{M}, \mathbf{Z}) \sim \eta}( \mathbf{X}^{-i} ; \Pi^{i}(\mathbf{X}) \given \mathbf{X}^i)   + \log k. 
\end{eqnarray*}
Putting these two calculations together we obtain that 
$$
\SIC_{\eta}(\Pi) < IC_\zeta(\Pi)  + O(n\log k), 
$$
and thus $ IC_\zeta(\Pi) >\SIC_{\eta}(\Pi) - O(n\log k) $, completing the proof. 

\end{proof}

We are now ready to prove our main theorem:

\begin{theorem}
For any $\delta > 0$, $n \geq 1$ and for $k = \Omega(\log n)$,
$$
IC_\zeta(\disj_{n,k}) = \Omega(n\cdot k) \hspace{0.2in} \mbox{and} \hspace{0.2in}
CC_\delta(\disj_{n,k}) = \Omega(n\cdot k). 
$$
\end{theorem}

\begin{proof}
The first part of the theorem follows from Theorem~\ref{lem:ICsetdisj} and Theorem~\ref{thm:sicdisj}.
The second part follows from the first part as well as the connection between communication complexity and information complexity
(Lemma~\ref{lem:cc-ic}).
\end{proof}

%% file: reductions.tex
\section{Lower Bound for Task Allocation}
\label{task}

In the task allocation problem, there are $k$ players and $n$ tasks.
Each player $i$ receives as input a set $X^i$ which specifies a subset of
tasks that it is capable of performing. The goal is for the players
to partition the tasks between them: each player $i$ must output
a subset of tasks such that each task is completed by exactly one player.
To make this problem feasible, we consider only inputs for which each
tasks has at least one player who is capable of performing it. Thus,
task allocation is a promise problem. We require 
that, at the end of the protocol, the coordinator knows which player is 
assigned to each task. 

%More precisely, the output is a collection of $n$-dimensional
%vectors $(y_1,\ldots,y_n)$ where each $y_i \in \{0,1\}^n$, such that:
%(i) for every $i \in [k]$, and $j \in [n]$, if $y_i[j]=1$, then $x_i[j]=1$;
%and (ii) for every $j \in [n]$, if $x_i[j]=1$ for some $i$, then it must
%be the case that $y_i[j]=1$ for exactly one player $i$.

Task allocation is a distributed one-shot variant of the well-known
$k$-server problem, where a centralized online algortihm assigns
tasks to $k$ servers, minimizing the total cost of servicing all tasks.
In the $k$-server problem, each (server, task) pair is associated
with a cost for having the server perform the task, and the
tasks arrive continually and must be assigned in an online manner.
In the task allocation problem, all tasks are given in the beginning, and all
have a cost of either 1 or infinity. Partitioning the tasks
between the players corresponds to finding a minimum-weight
assignment of tasks to servers. 
Task allocation is also closely related to the problem
of finding a rooted spanning tree in directed broadcast networks
\cite{DruckerKuhnOshman}.

Drucker, Kuhn and Oshman~\cite{DruckerKuhnOshman} showed 
tight communication complexity lower bounds for 
the two player task allocation problem.
%Recently, a reduction to 2-player set disjointness was given by Noga Alon.
In this section, we generalize this lower bound to the 
$k$ player setting by showing an $\Omega(nk)$ lower
bound for task allocation in the message passing model.

Our reduction is similar in spirit to the reduction, due to Noga Alon~\cite{nogaalon}, 
from two-party set disjointness to the promise task allocation problem in the two-player case.

\hide{

\paragraph{From Set Disjointness to Task Allocation: Removing the Promise.} 
We obtain our lower bound on task allocation via a reduction 
from set disjointness in the coordinator model.

Our starting point is an efficient randomized reduction from two-party 
set disjointness to the promise task allocation problem in the two player 
case, due to Noga Alon~\cite{nogaalon}. 

Observe that sets $X^1,\ldots,X^k$ are disjoint if and only if
the sets $[n]-X^1$, $[n]-X^2$, \ldots, $[n]-X^k$ together cover $[n]$.
This gives us the following protocol for two-player set disjointness.

Given an instance $(X^1, X^2)$ of the set disjointness problem,
the two parties first run the task allocation protocol on
inputs $[n]-X^1$ and $[n]-X^2$ respectively.
Suppose that the protocol returns the allocation
$Y^1 \subseteq [n]$ to the first player, and $Y^2 \subseteq [n]$
to the second player.
The players check that $Y^1$ and $Y^2$ are consistent with
their respective inputs. That is, Player 1 checks that 
$Y^1 \subseteq [n]-X^1$ and similarly Player 2 checks that
$Y^2 \subseteq [n]-X^2$.
If one of the checks does not pass, then they declare that
$X^1$ and $X^2$ are not disjoint and abort.
Otherwise (if the consistency checks pass), they 
run an equality protocol on $[n]-Y^1$ and $Y^2$
and declare that $X^1$ and $X^2$ are disjoint if and only if the
equality protocol accepts.

For correctness, first consider the case where 
the sets are disjoint. Because the sets are disjoint,
the promise that the inputs to the task allocation protocol
cover $[n]$ is satisfied, and thus the protocol returns a
legal task allocation, $Y^1$ and $Y^2$, and thus the players' consistency checks
will pass. Because $Y^1$ and $Y^2$ forms a disjoint cover of $[n]$,
$[n]-Y^1$ will equal $Y^2$, and thus the
equality protocol will accept with high probability and 
the players will declare that the sets are disjoint with high
probability.

The second case is when the sets are not disjoint.
Let $j$ be a coordinate such that $X^1_j = X^2_j =1$.
In this case, the promise is not satisfied. If a consistency check
does not pass, then the players output the correct answer (not disjoint).
Otherwise, suppose both consistency checks pass. In this case it
must be that their allocations $Y^1, Y^2$ satisfy $Y^1_j = Y^2_j = 0$,
and thus $[n]-Y^1$ does not equal $Y^2$, and thus with high
probability the protocol declares not disjoint.

Note that in the two-player case, 
set disjointness is known to require randomized complexity $\Omega(n)$, and 
thus this gives an alternative proof of a linear lower bound for two-player
Task Allocation.
}

%We note that a related problem, non-promise task allocation
%can be shown to be hard via a reduction 
%from the bit-wise OR problem to
%task allocation in the coordinator model.
%Recall that in the bitwise OR problem, each player is given an $n$-bit string
%and they want to compute the bitwise OR of their strings.
%Suppose that the inputs of the players are $x^1,\ldots,x^k$.
%They first simulate the task allocation protocol on their inputs
%to obtain sets $S^1,\ldots,S^k$, $S^i \leq x^i$.
%Clearly from these sets the coordinator can recover
%the bitwise-OR since for any index $j$ the OR of the $j$th
%bit is 1 if and only if there exists a set $S^l$ such that
%$S^l_j =1$.
%
%Phillips, Verbin and Zhang prove an $O(nk)$ lower bound for the bitwise OR function in the
%coordinator model ad thus we obtain the same lower bound for multiplayer task allocation.
%Thus we have the following theorem.

%\begin{theorem}
%Task allocation requires $O(nk)$ complexity in the
%coordinator model, the private message model, and the pay-per-view model.
%\end{theorem}

%It is also worth noting that Task Allocation efficiently
%reduces to Promise Task Allocation plus bit-wise OR (again in the
%coordinator model).

\begin{theorem}
There is a reduction from $k$-party Set Disjointness to $k$-party (Promise) 
Task Allocation with an overhead of $O(n\log n + k)$ bits. That is, given 
a task allocation protocol that communicates $C_{\mathsf{TA}}(n,k)$ bits, 
there is a protocol for set disjointness that communicates
\[C_{\mathsf{SD}}(n,k) = 
C_{\mathsf{TA}}(n,k) + O(n\log n + k) \] bits.
Thus, for large enough $n$ and large enough $a$,
if $k \geq  a \log n$, then
the communication complexity of Task Allocation in the coordinator model
is $\Omega(nk)$.
% assuming that Set Disjointness
%requires $O(nk)$ complexity.
\end{theorem}

\begin{proof}
We now give the reduction from multiplayer set disjointness
to multiparty (promise) task allocation in the coordinator model.
%As with Alon's reduction in the two-party setting, one can 
%think of this reduction as ``removing the promise''. 
Let the input to the set disjointness problem be $X^1,\ldots,X^k$. 
Define $Y^i = [n] - X^i$.
%The players compute
As before, note that 
\[ \cap_{i=1}^k X^i  \mbox{ is empty} \hspace{0.2in} \mbox{\em if and only if} \hspace{0.2in} \cup_{i=1}^k Y^i = [n] \]
%The inputs are disjoint if and only if the sets
%$[n]-X^i$ forms a partition of $[n]$.

The players simulate the protocol for task allocation
on the inputs $Y^i$, and the coordinator gets the output 
$(Z^1,Z^2,\ldots,Z^k)$ where $Z^i$ is the set of tasks 
that player $i$ is expected to complete.\footnote{Our definition of the coordinator model postulates that the 
coordinator learns the result at the end of the protocol. A similar reduction works in the case where the players learn their respective outputs,
namely each player $i$ learns $Z^i$, the set of tasks that he is expected to complete.}  %suppose the output is $Y^1,\ldots,Y^k$.
The protocol then proceeds as follows:

\begin{enumerate}

\item The coordinator checks that the sets $Z^{1}, \ldots, Z^k$ form a partition of $[n]$. 
If not, the coordinator outputs ``not disjoint'' and halts. If the check passes, proceed to the
next step. 

\item The coordinator sends $Z^i$ to player $i$. Each player $i$ checks whether $Z^i \subseteq Y^i$.
If $Z^i \nsubseteq Y^i$, player $i$ sends a ``not disjoint'' message to the coordinator (and if $Z^i \subseteq Y^i$,
it sends an ``OK'' message.) If the coordinator receives a ``not disjoint'' message from any player,
it outputs ``not disjoint'' and halts. Otherwise, it outputs ``disjoint'' and halts. 

\end{enumerate}

Assume that $(X^1,\ldots,X^k)$ is a YES instance of set disjointness, namely that $\cap_{i=1}^k X^i = \varphi$. 
Then, $\cup_{i=1}^k Y^i = [n]$, and the input $(Y^1,\ldots,Y^k)$ satisfies the promise to the task allocation 
problem. By the correctness of the task allocation protocol, the output $(Z^1,\ldots,Z^k)$ is a valid allocation
of the tasks, namely, $(Z^1,\ldots,Z^k)$ forms a partition of the universe $[n]$, and $Z^i \subseteq Y^i$ for each $i \in [k]$.
Thus, the coordinator will output ``disjoint'' in the above protocol.

On the other hand, if $(X^1,\ldots,X^k)$ is a NO instance of set disjointness, then we know that 
$\cup_{i=1}^k Y^i \neq [n]$ and either of the following
two events happen:

\begin{itemize}

\item for some $i$, $Z^i \nsubseteq Y^i$; or
\item $\cup_{i=1}^k Z^i \neq [n]$

\end{itemize}

Player $i$ will detect the first of these two cases, and the coordinator will detect the second.
In either case, the coordinator will output ``not disjoint''.

As for the complexity of the protocol, the coordinator runs step $2$ of the protocol only 
if step $1$ passes, namely if $(Z^1,\ldots,Z^k)$ forms a partition of $[n]$. In this case, the 
overhead of the protocol is $O(n\log n + k)$ bits.  Since the communication complexity of set 
disjointness in the coordinator model is $\Omega(nk)$, so is the communication complexity of task allocation.
\end{proof}

\hide{
As in the two player case, each player $i$ performs a consistency
check on the answer, namely they each check that $Y^i \subseteq Z^i$. If the
consistency check fails,  and they send  ``not disjoint" 
message to the coordinator if their 
consistency checks fail.

Otherwise, they do the following  in order to check that the $Y^i$'s 
cover the universe $[n]$:

\begin{itemize}

\item They first compute the  total number of elements in all the sets $Y^{i}$, that is
they compute
\[ \mathsf{count} = \sum_{i=1}^{k} |Y^{i}| \]
This is done by having each player $i$ send $Y^{i}$ to the coordinator. 
If $\mathsf{count} \neq n$, the coordinator outputs ``not disjoint''.

\item If $\mathsf{count} = n$, then the players send their sets $Y^i$ to the coordinator
who computes $\bigcup_{i=1}^k Y^{i}$.

\end{itemize}

If 
}

\hide{
\Vnote{not a partition.}

Otherwise, they want to check whether or not their output
forms a partition of $[n]$.
The players compute the total number of 1's in the
union of their sets $S^i$, as follows.
Player 1 sends their elements to Player 2, who then sends
the total number of 1's in the union of $S^1$ and $S^2$ to Player 3,
and so on.
The total complexity is $O((k+n)polylog(n,k))$
They declare "disjoint" if and only if the size of the union is
exactly $n$.

\end{proof}
}

\hide{

\subsection{Lower Bounds for Spanning Tree}

The distributed spanning tree problem (ST) is a fundamental primitive in distributed
computing. 
The input is a strongly connected directed graph $G=(V,E)$.
This graph is viewed as a network where vertices $V$ are independent processes,
and edges $E$ are directed communication links between pairs of processes. At the start
of the algorithm, for each $i \in V$, vertex $i$ is given a Boolean vector $x^i$ where $x^i_j=1$ if and only if
there is a directed edge from vertex $j$ to vertex $i$ in $G$.
(That is, at the start of the algorithm, each vertex knows which other vertices
can talk to him/her.)
The goal is to output a rooted spanning tree where all edges in the tree
point up towards the root.
At the end of the algorithm, each vertex $i$ should know which of its inedges
belong to the spanning tree and which do not.

\begin{theorem}
The distributed spanning tree problem requires $\Omega(n^2)$ communication in the
message passing model.
\end{theorem}

\begin{proof}
We show how to solve task allocation for $k=n$ players and $n$ tasks in the 
coordinator model (and hence in the message passing
model) with an efficient protocol for spanning tree.

Before doing this, we will consider a restriction of the task allocation
problem where in addition to the promise that every task is assigned to
at least one player, we are also guaranteed that every player gets at least
one task. Call this restricted problem the onto task allocation problem.
It is not hard to see that the onto task allocation problem is as hard
as the task allocation problem in the coordinator model. Given an instance of task allocation,
during the first round, the players tell the coordinator whether or not they
are in the game. (A player is in the game only if he/she is assigned a task.) 
The coordinator counts the number of players, $k' \leq k$, who are
in the game, and sends this number to each player, for a total of $k \log k$ bits.
Then the players who are in the game follow the onto task allocation protocol
with $k'$ players, and $n$ tasks.

Now we will show how to solve the onto task allocation problem for $n$ players
and $n$ tasks in the coordinator model, with a protocol for spanning tree
on a graph of size $2n+1$.
Given an input $X^1, X^2,\ldots, X^k$ to onto task allocation,
where $|X^i| =n$, we construct an instance of the spanning tree problem
on a graph with $kn+1$ vertices.
The graph will be a layered graph; at the bottom layer there are $n$
vertices, $t_1,\ldots,t_n$ corresponding to the $n$ tasks, at the middle layer there are $k$
vertices $x_1,\ldots,x_k$ corresponding to the $k$ players, and at the top layer there is
one additional vertex, $v$.
The edges are as follows. First, there is an edge from $v$ to $t_i$ for all $i \leq n$;
Secondly, there is an edge from $x_j$ to $v$ for all $j \leq k$;
and thirdly, there is an edge $(t_i,x_j)$ if and only if
$X^j(i)=1$. (That is, if and only if player $j$ is assigned task $i$.)

It is easy to see that the resulting graph is strongly connected.
To see this, first note that for every $i \neq i'$, there is a path
from $t_i$ to $t_{i'}$.
Secondly,  for every $i \leq n$ and $j \leq k$ there is
a path from $t_i$ to $x_j$. This is because since the task allocation
instance is onto, there is at least one $i'$ such that
$t_{i'}$ is connected to $x_j$. Thus, to go from $t_i$ to $x_j$, we 
can go from $t_i$ to $t_{i'}$, and then from $t_{i'}$ to $x_j$.

Now the players need to simulate the spanning tree protocol in the coordinator model.
\end{proof}
}